\documentclass[10pt,journal,cspaper,compsoc]{IEEEtran}

\usepackage{eso-pic}
\usepackage[pdftex]{graphicx}
\usepackage{tikz}
\usepackage{pgf}
\usepackage{lscape}
\usepackage{amssymb}
\usepackage{fancyvrb}
\usepackage{amsfonts}
\usepackage{xfrac}
\usepackage{mathtools}
\usepackage{verbatim}
\usepackage{algorithm}
\usepackage{algcompatible}
\usepackage{amsmath}
\usepackage{multirow}
\usepackage{epstopdf}
\usepackage{amsthm}
\usepackage[justification=centering, font={footnotesize}]{caption}
\usepackage{subfigure}
\usepackage{soul}
\usepackage{proof}
\usetikzlibrary{arrows,automata}
\usetikzlibrary{shapes}

\title{Accurate and Efficient Modeling of 802.15.4 Unslotted CSMA/CA through Event Chains Computation}
\makeatletter
\def\markboth#1#2{\def\leftmark{\@IEEEcompsoconly{\sffamily}\MakeUppercase{\protect#1}}%
\def\rightmark{\@IEEEcompsoconly{\sffamily}\MakeUppercase{\protect#2}}}
\makeatother

\author{Domenico De Guglielmo,
        Francesco Restuccia,
        Giuseppe Anastasi,
        Marco Conti,
        and Sajal K. Das %
\IEEEcompsocitemizethanks{
\IEEEcompsocthanksitem D. De Guglielmo and G. Anastasi are with the Department of Information Engineering,
University of Pisa, Italy (email: \{d.deguglielmo, g.anastasi\}@iet.unipi.it).
\IEEEcompsocthanksitem F. Restuccia and S.K. Das are with the Department of Computer Science, Missouri University of Science and Technology, Rolla, MO,  United States (e-mail: \{frthf, sdas\}@mst.edu). 
\IEEEcompsocthanksitem M. Conti is with the Institute of Informatics and Telematics (IIT), National Research Council, Pisa, Italy (e-mail: marco.conti@iit.cnr.it).
}
\thanks{\fbox{\parbox{\dimexpr\columnwidth-\fboxsep-\fboxrule\relax}{
\copyright 2016 IEEE. Personal use of this material is permitted.  Permission from IEEE must be obtained for all other uses, in any current or future media, including reprinting/republishing this material for advertising or promotional purposes, creating new collective works, for resale or redistribution to servers or lists, or reuse of any copyrighted component of this work in other works.}}}
}

\begin{document}

\markboth{}%
{Shell \MakeLowercase{\textit{et al.}}: Bare Demo of IEEEtran.cls for Journals}

\IEEEcompsoctitleabstractindextext{%
\begin{abstract}
Many analytical models have been proposed for evaluating the performance of event-driven 802.15.4 Wireless Sensor Networks (WSNs), in \emph{Non-Beacon Enabled (NBE)} mode. However, existing models do not provide accurate analysis of large-scale WSNs, due to tractability issues and/or simplifying assumptions. In this paper, we propose a new approach called \emph{Event Chains Computation (ECC)} to model the unslotted CSMA/CA algorithm used for channel access in NBE mode. ECC relies on the idea that outcomes of the CSMA/CA algorithm can be represented as \emph{chains} of \emph{events} that subsequently occur in the network. Although ECC can generate all the possible outcomes, it only considers chains with a probability to occur greater than a pre-defined threshold to reduce complexity. Furthermore, ECC parallelizes the computation by managing different chains through different threads. Our results show that, by an appropriate threshold selection, the time to derive performance metrics can be drastically reduced, with negligible impact on accuracy. We also show that the computation time decreases almost linearly with the number of employed threads. We validate our model through simulations and testbed experiments, and use it to investigate the impact of different parameters on the WSN performance, in terms of delivery ratio, latency, and energy consumption.
\end{abstract}
\begin{keywords}
Wireless Sensor Networks, IEEE 802.15.4, Non-Beacon Enabled mode, unslotted CSMA/CA, Performance Analysis.
\end{keywords}
}

\maketitle

\IEEEdisplaynotcompsoctitleabstractindextext
\IEEEpeerreviewmaketitle

\section{Introduction}
The IEEE 802.15.4 standard is currently the reference communication technology for wireless sensor networks (WSNs) \cite{ieee154} and is expected to be a major enabling technology for the future Internet of Things (IoT) \cite{abire}. The 802.15.4 standard addresses the physical and medium access control (MAC) layers of the networking stack and is complemented by the ZigBee specifications \cite{zigbee} that cover the network and application layers.

 Due to the wide range of potential applications, ranging from environmental monitoring to smart grids, from urban mobility to  healthcare, from industrial applications to mobile ticketing and assisted driving, and so on, the 802.15.4 standard leverages a number of different access methods to cope with different quality of service (QoS) requirements. In particular, the standard defines a Beacon-Enabled (BE) mode and a Non-Beacon-Enabled (NBE) mode. The BE mode relies on a periodic superframe bounded by beacons, which are special synchronization messages generated by coordinator nodes. Each sensor node waits for the reception of a beacon, and then, starts transmitting its data packets using a slotted carrier sense multiple access with collision avoidance (CSMA/CA) algorithm. Conversely, in the NBE mode there is no superframe; nodes are not synchronized, and use an unslotted CSMA/CA algorithm for packet transmissions. 

The performance of 802.15.4 WSNs has been thoroughly investigated in the past \cite{ramach}-\cite{petrova}. However, most of the previous studies have focused on the BE mode and, hence, they have considered the slotted CSMA/CA algorithm. Although the NBE mode is the most suitable access method for applications generating sporadic and/or irregular traffic (e.g. event-driven WSN applications and upcoming IoT applications), significantly less attention has been devoted to it. 

The most challenging aspect in analyzing the 802.15.4 unslotted CSMA/CA algorithm lies in its remarkable complexity, mainly due to its random nature and the multitude of parameters regulating its behavior. Nevertheless, having an accurate (yet tractable) model is imperative to investigate the QoS that can be provided to applications. The majority of previous models of unslotted CSMA/CA algorithm assume that packets are generated according to a specific stochastic distribution (e.g. Poisson). However, this assumption is not valid for WSN applications where nodes generate traffic according to an \emph{event-driven} pattern (e.g. query-based applications, neighbor discovery, data aggregation). Although some works \cite{bur, griba} have proposed analytical models of the unslotted CSMA/CA algorithm for event-driven WSNs, literature still lacks models that are both \emph{accurate} and \emph{tractable}. For this reason, in this paper we present a new model of unslotted CSMA/CA, for event-driven WSNs, that is both accurate and tractable. We use it to derive performance metrics of interest such as delivery ratio, delay, and energy consumption. In order to deal with the significant complexity of the algorithm, we use an approach called {\em Event Chains Computation} (ECC). 

ECC relies on the idea that \emph{outcomes} of the CSMA/CA algorithm can be represented as a sequence (\emph{chain}) of transmissions (\emph{events}) that subsequently occur in the network. ECC is able to iteratively build all the possible sequences of events that can be experienced by sensor nodes while transmitting a data packet. However, to reduce complexity, only the event chains whose probability to occur is above a predefined threshold are considered. By appropriate selection of the threshold, it is possible to reduce the model complexity and, hence, the computational resources needed to calculate the performance metrics, with a limited loss of accuracy. In addition, it is possible to parallelize the algorithm, which further reduces the computation time.

To summarize, this paper makes the following contributions. \vspace{0.1cm}

$\bullet$ We introduce the ECC modeling approach that leverages both the analysis of the most likely events and parallel computation, to reduce drastically the model computation time. 

$\bullet$ We use ECC to derive performance metrics of interest such as delay, energy consumption, and packet delivery probability. 

$\bullet$ We validate our model through simulation and experiments in a real testbed. Also, we analyze the performance of the unslotted 802.5.4 CSMA/CA algorithm as a function of different operating parameters. Our results demonstrate that the ECC approach allows to reduce drastically the computation time, while guaranteeing a good accuracy for the computed performance metrics.\vspace{0.1cm}

The paper is organized as follows. Section 2 presents related work. Section 3 describes the unslotted 802.15.4 CSMA/CA algorithm. Section 4 presents the model assumptions. Section 5 and 6 detail the ECC algorithm and derive performance metrics. Section 7 presents the obtained results. Finally, Section 8 draws conclusions.

\section{Related Work}

The 802.15.4 MAC protocol has been thoroughly investigated through analysis \cite{ramach}-\cite{griba}, simulations \cite{singh, ana} and real experiments \cite{ana, petrova}. However, most of these studies have focused on the BE mode (i.e., slotted CSMA/CA), while significantly less attention has been devoted to the NBE mode. In this paper, we focus on the latter mode and, hence, on the unslotted CSMA/CA, hereafter referred to as CSMA/CA for brevity.





One of the first models for CSMA/CA was due to Goyal {\em et al.} \cite{goya}, who proposed a stochastic model assuming that packet inter-arrival times follow an exponential distribution, and considering the effect of packet retransmissions. More recently, Di Marco et al. \cite{park2} analyzed the CSMA/CA algorithm in single and multi-hop scenarios, through an accurate model based on Discrete Time Markov Chains (DTMCs). In \cite{dimarco2014} the authors improved the work in \cite{park2} by proposing an extended model considering the effect of a fading channel.

In \cite{park2}-\cite{goya} it is assumed that sensor nodes generate data packets according to a Poisson distribution. However, this assumption does not apply to a large number of WSN scenarios where sensor nodes typically follow an \emph{event-driven} reporting paradigm \cite{cao2015}. In this paper, we consider event-driven applications and assume that, when an event is detected, all (or a large number of) sensor nodes in the network start reporting data simultaneously (the event-driven paradigm can be easily extended to model periodic traffic as well). A similar scenario is considered in \cite{cao2015} where the authors provide an accurate model of the 802.15.4 CSMA/CA (that considers both the case with and without retransmissions) and validate it through extensive simulations. We point out that, differently from this paper, \cite{cao2015} focuses on the slotted version of 802.15.4 CSMA/CA and this completely changes the analytical model. In addition, no real experiments are provided to validate the model.

Regarding event-driven WSNs and the unslotted CSMA/CA, the closest to our work are \cite{bur}-\cite{griba}. In \cite{bur} the authors derive the packet latency and delivery ratio experienced by $N$ sensor nodes that attempt to transmit a single packet simultaneously to the sink node. The work was further extended in \cite{martelli2014} by taking into account the hidden node problem using the theory of stochastic geometry. However, conversely from us, both \cite{bur} and \cite{martelli2014} do not consider the effects of acknowledgements and retransmissions. In \cite{griba}, Gribaudo {\em et al.} provide a very accurate and complete analytical model of CSMA/CA using stochastic automata networks (SANs) \cite{stoc}. The analysis is mainly aimed at deriving the packet delay distribution and on-time delivery ratio (percentage of packets received by the sink node within a pre-defined threshold). As in \cite{bur} and \cite{martelli2014}, the analysis of the energy consumption of sensor nodes is neglected. Furthermore, although the use of SANs makes the analysis very accurate, it also raises serious complexity issues. In fact, the analysis in \cite{griba} is limited to WSNs with a low number of nodes (i.e. less than 6). This is because the size of the WSN global descriptor (i.e., a matrix) increases exponentially with the network size. Hence, the computation time and memory needed to solve the model increases accordingly.

Likewise the model proposed in \cite{griba}, our analysis is very accurate. It considers acknowledgements and retransmissions, and no over-simplifying assumptions are made on the CSMA/CA algorithm. In terms of performance metrics, we derive delivery ratio, packet latency, \emph{and} energy consumption of sensor nodes. However, our most important contribution is to provide an accurate analytical model of CSMA/CA which is able to analyze WSNs with a large number of nodes. This is because we do \emph{not} use a matrix-based analytical model (like DTMCs or SaNs). Instead we undertake an approach called \emph{Event Chains Computation} (ECC), that makes the analysis very accurate yet computationally tractable. As we will show in Section 6, ECC is scalable and, unlike  previous techniques, is particularly suitable for parallelization, due to its intrinsic concurrent structure. This contributes to drastic reduction of computation time. To the best of our knowledge, ours is the first accurate analytical model of the (unslotted) CSMA/CA algorithm for event-driven scenarios investigating the performance of WSNs with a large number of sensor nodes. In addition, we present a comparison of analytical and experimental results. 
\section{CSMA/CA Algorithm}
According to the IEEE 802.15.4 standard, in WSNs operating in the NBE mode, sensor nodes must associate with a coordinator node and send their packets to it, using the (unslotted) CSMA/CA algorithm. Unlike regular sensor nodes, coordinator nodes are energy-unconstrained devices that form a higher-level network aimed at forwarding data to the final destination. Specifically, coordinator nodes are always on and, thus, sensor nodes are allowed to start a packet transmission at any time. In addition, no synchronization is required. Below, we provide a short description of the CSMA/CA algorithm used by sensor nodes to transmit data to their coordinator node. 

Upon receiving a data packet, the MAC layer at the sensor node performs the following steps.
\begin{enumerate}
 \item A set of state variables is initialized, namely the number of backoff stages carried out for the on-going transmission (NB = 0) and the backoff exponent (BE = $\mbox{\em{macMinBE}}$).
 \item A random backoff time is generated and used to initialize a timer. The backoff time is obtained by multiplying an integer number uniformly distributed in $[0,\ 2^{BE-1}]$ by the the duration of the \emph{backoff-period} ($D_{bp}$). As soon as the timer expires the algorithm moves to step 3.
 \item A Clear Channel Assessment (CCA) is performed to check the state of the wireless medium.
 \begin{enumerate}
  \item If the medium is free, the packet is immediately transmitted.
  \item If the medium is busy, state variables are updated as follows: NB = NB+1 and BE = min(BE+1,$\mbox{\em{macMaxBE}})$. If the number of backoff stages has exceeded the maximum allowed value (i.e. NB $>$ $\mbox{\em{macMaxCSMABackoffs}}$), the packet is dropped. Otherwise, the algorithm falls back to step 2.
 \end{enumerate}
\end{enumerate}

The CSMA/CA algorithm supports an optional retransmission scheme based on acknowledgements and timeouts. When the retransmissions are enabled, the destination node must send an acknowledgement upon receiving a correct data packet. On the sender side, if the acknowledgment is not (correctly) received within a pre-defined timeout, the packet is retransmitted, unless the maximum number of allowed retransmissions ({\em macMaxFrameRetries}) has been reached. Otherwise, the packet is dropped.

\section{Model Assumptions}
In the following, we focus on the communication between sensor nodes and the coordinator node they are associated with. We assume that there are $N$ nodes associated with the considered coordinator node. We refer to event-driven applications in which \emph{all} sensor nodes start transmitting a data packet simultaneously to their coordinator, to report a detected physical event. This is one of the most challenging scenarios in terms of performance and energy consumption. We assume that each sensor node is in the carrier sensing range of each other and there are no obstacles in the sensing field. This assures that the hidden node problem never arises. In addition, we assume that the time between two consecutive physical events $ph_1$ and $ph_2$ is long enough to assure that the execution of the CSMA/CA algorithm, started by sensor nodes to report $ph_1$, is surely terminated before $ph_2$ occurs. Hence, in the following, we focus on a single physical event $ph$. We indicate as $t=0$ the time at which $ph$ occurs and, hence, the time at which all the $N$ nodes start executing the CSMA/CA algorithm. Finally, we make the following assumptions:
\begin{itemize}
\item Each sensor node transmits a single packet to report the detected event $ph$.
\item Data packets transmitted by different sensor nodes have the same size. In particular, we assume that the packet size is such that the corresponding transmission time can assume a value $D_{tx}$ such that $D_{tx}=D_{max} - k \cdot D_{bp},\mbox{ for } k \ge 0$, where $D_{max}$ is the time required to transmit a maximum-size packet (133 bytes) and $D_{bp}$ is the duration of the backoff period.
\item The communication channel is ideal, i.e., data/acknowledgement packets are never corrupted, or lost, due to transmission errors.
\end{itemize}

\section{Event Chains Computation}
The CSMA/CA algorithm is a random access algorithm whose goal is to minimize the probability of collision between packets transmitted by different sensor nodes. Due to its random nature, different executions of the algorithm can yield completely different outcomes. For instance, if we consider $N$ nodes simultaneously transmitting a single data packet to their coordinator node, a run of the algorithm can result in a transmission schedule such that all the $N$ data packets are successfully transmitted to the coordinator (and, hence, 100\% reliability is achieved) while another run can result in no successful transmissions at all (e.g. due to repeated collisions). Obviously, different outcomes have, in general, different probabilities to occur. 

The ECC algorithm is able to generate all the possible \emph{outcomes} an execution of CSMA/CA algorithm can yield, and the corresponding probabilities. However, to reduce the complexity of the analysis, it is possible to instruct the ECC algorithm to generate only the outcomes having a probability to occur greater than or equal to a certain threshold $\theta\ (0\le \theta < 1)$. The set of possible outcomes produced by the algorithm is then used to calculate the performance metrics of interest such as delivery ratio, latency and energy consumption.
\begin{figure}[htbp]
	\centering
		\includegraphics[scale=0.50]{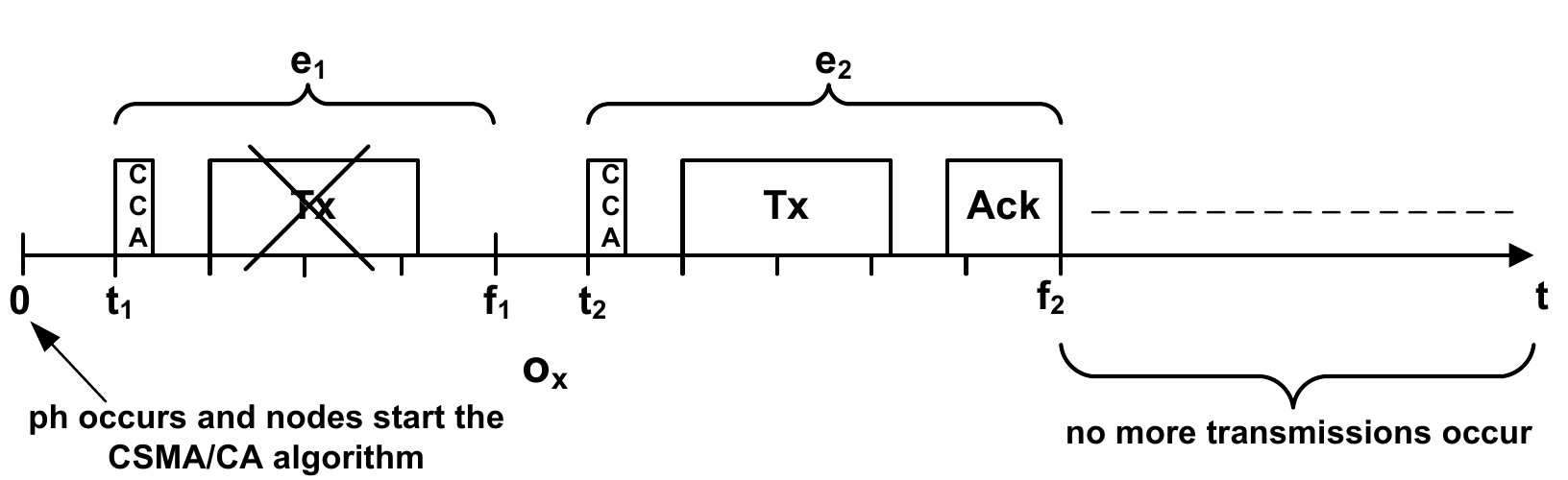}
	\caption{$o_x$, a possible outcome of the CSMA/CA.}
	\label{fig:ox}
\end{figure}

The ECC algorithm is based on the observation that an outcome of the CSMA/CA execution can always be represented as a series (\emph{chain}) of successful/failure transmissions (\emph{events}) occurring subsequently in the network. Figure \ref{fig:ox} shows a possible outcome $o_x$ of the CSMA/CA algorithm representing the case when  a transmission failure (time $t=t_1$) followed by a successful transmission (time $t=t_2$) occur. Please note that time $t=0$ is the time at which all the nodes start the CSMA/CA execution and that, in this specific example, no more transmissions occur in the network after those shown in the figure.

Let us indicate as $e_1$ and $e_2$, respectively, the transmission failure and the successful transmission depicted in Figure \ref{fig:ox}. Then, the probability of outcome $o_x$ is:
\begin{equation}
\mathbb{P}\{o_x\}=\mathbb{P}\{ e_1 \wedge e_2 \wedge no\_txs\}
\label{eq:pox}
\end{equation}
where $no\_txs$ indicates that no more events (successful/failure transmissions) occur in the network after $e_2$. 

By recursively applying the Bayes' theorem, eq. \ref{eq:pox} can be rewritten as:
\begin{equation} 
\begin{split}
\mathbb{P}\{o_x\} & = \mathbb{P}\{ e_1 \wedge e_2 \wedge no\_txs\} \\
 & = \mathbb{P}\{e_1 \wedge e_2\} \mathbb{P}\{no\_txs\ |\ e_1 \wedge e_2\} \\
 & = \mathbb{P}\{e_1\} \mathbb{P}\{e_2\ |\ e_1\} \mathbb{P}\{no\_txs\ |\ e_1 \wedge e_2\} 
\end{split}
\end{equation}
More generally, the probability of an outcome $o_i$, representing the series of events $e_1,\ e_2,\ ...,\ e_n$, can be calculated as follows:
\begin{equation} 
\begin{split}
\mathbb{P}\{o_i\} & = \mathbb{P}\{ e_1 \wedge e_2 \wedge...\wedge e_n \wedge no\_txs\} \\
 & = \mathbb{P}\{ e_1 \wedge\ ...\wedge e_n\} \mathbb{P}\{no\_txs\ |\ e_1 \wedge...\wedge e_n\} \\
 & = \mathbb{P}\{e_1\} \mathbb{P}\{e_2\ |\ e_1\} \mathbb{P}\{e_3\ |\ e_1\wedge e_2\} \cdot ...\\
 & \cdot\ \mathbb{P}\{no\_txs\ |\ e_1 \wedge...\wedge e_n\}
\label{eq:poigen}
\end{split}
\end{equation}
Eq. \ref{eq:poigen} suggests that, in order to calculate the probability of outcome $o_i$, $n+1$ different steps have to be performed. First, $\mathbb{P}\{e_1\}$ has to be computed, i.e. the probability that $e_1$ is the first event occurring in the network. Then, at each subsequent step $k,\ 2\le k \le n$, the probability $\mathbb{P}\{e_k\ |\ e_1 \wedge ... \wedge e_{k-1}\}$ that event $e_k$ occurs, given that all the previous $k-1$ events have occurred, has to be derived. Finally, $\mathbb{P}\{no\_txs\ |\ e_1 \wedge...\wedge e_n\}$ has to be calculated, i.e. the probability that no other events will occur in the network after $e_n$. The ECC algorithm follows exactly these $n+1$ steps to compute the probability of an outcome.

Now, by means of a simple example, we give an overview of the actions performed by ECC to generate all the possible outcomes of a CSMA/CA execution and the corresponding probabilities.
\begin{figure}[htbp]
	\centering
		\includegraphics[scale=0.60]{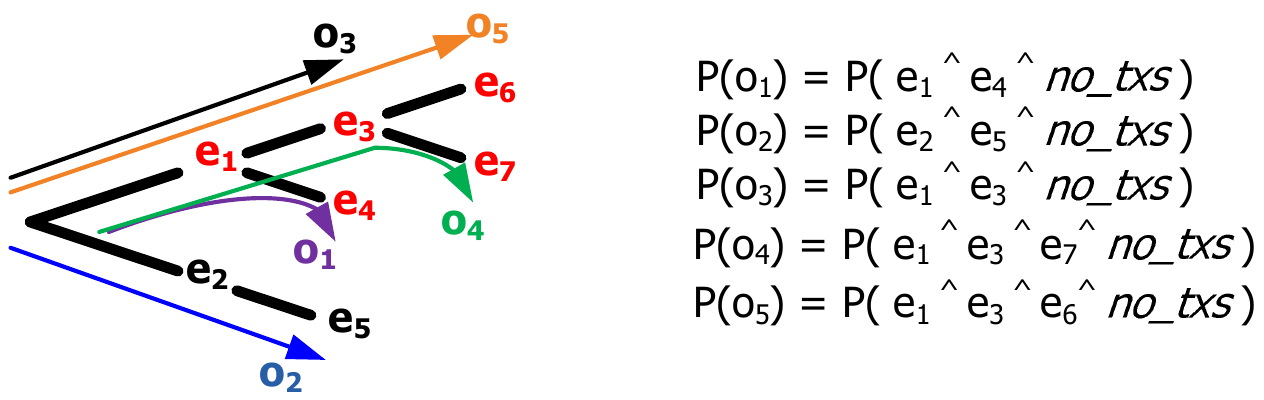}
	\caption{Possible outcomes of the CSMA/CA.}
	\label{fig:graphoutcomes}
\end{figure}

Figure \ref{fig:graphoutcomes} represents a case where a CSMA/CA execution can produce five possible outcomes, namely $o_1,\ o_2,\ o_3,\ o_4,\ o_5$. Initially (i.e. just after time $t=0$), two possible events can occur, namely $e_1$ and $e_2$. The ECC algorithm calculates that, with probability $\mathbb{P}\{e_1\}$, $e_1$ is the first event to occur in the network while, with probability $\mathbb{P}\{e_2\}$, $e_2$ occurs. Then, the algorithm checks if there are cases in which no other events occur in the network after $e_1$ or $e_2$ have occurred, i.e. if there are outcomes of the algorithm composed only of event $e_1$ or $e_2$. To this end, both $\mathbb{P}\{no\_txs\ |\ e_1\}$ and $\mathbb{P}\{no\_txs\ |\ e_2\}$ are calculated. Since in this example there are no outcomes terminating with $e_1$ or $e_2$, both $\mathbb{P}\{no\_txs\ |\ e_1\}$ and $\mathbb{P}\{no\_txs\ |\ e_2\}$ are equal to $0$. 

ECC then derives the events that can occur after $e_1$ or $e_2$. It discovers that, with probability $\mathbb{P}\{e_3\ |\ e_1\}$, $e_3$ will follow $e_1$, while with probability $\mathbb{P}\{e_4\ |\ e_1\}$, $e_4$ will follow $e_1$. Also, event $e_5$ is the only event that can occur in the netwok after $e_2$, i.e. $\mathbb{P}\{e_5\ |\ e_2\}=1$. As before, ECC checks if there are cases in which no other events occur in the network after $e_3$, $e_4$ or $e_5$ by computing $\mathbb{P}\{no\_txs\ |\ e_1\wedge e_3\}$, $\mathbb{P}\{no\_txs\ |\ e_1\wedge e_4\}$, $\mathbb{P}\{no\_txs\ |\ e_2\wedge e_5\}$. It discovers that both $\mathbb{P}\{no\_txs\ |\ e_1\wedge e_4\}$ and $\mathbb{P}\{no\_txs\ |\ e_2\wedge e_5\}$ are equal to $1$, since no events can occur after $e_4$ and $e_5$, while $0<\mathbb{P}\{no\_txs\ |\ e_1\wedge e_3\}<1$, i.e. there are cases in which no other events occur in the network after $e_1$ and $e_3$. Thus, the algorithm stores three different outcomes namely $o_1$, $o_2$, $o_3$ and the corresponding probabilities $\mathbb{P}\{o_1\},\ \mathbb{P}\{o_2\},\ \mathbb{P}\{o_3\}$ given by eq. \ref{eq:poigen}. 

Since $\mathbb{P}\{no\_txs\ |\ e_1\wedge e_3\}$ is less than $1$, there are also cases in which other events may occur in the network after $e_1$ and $e_3$. Specifically, $e_6$ occurs with probability $\mathbb{P}\{e_6\ |\ e_1\wedge e_3\}$ while $e_7$ occurs with probability $\mathbb{P}\{e_7\ |\ e_1\wedge e_3\}$. Also, since no other events can occur after $e_6$ and $e_7$ both $\mathbb{P}\{no\_txs\ |\ e_1\wedge e_3\wedge e_6\}$ and $\mathbb{P}\{no\_txs\ |\ e_1\wedge e_3\wedge e_7\}$ are equal to $1$. Hence, the algorithm stores outcomes $o_4$ and $o_5$ with $\mathbb{P}\{o_4\}$ and $\mathbb{P}\{o_5\}$ calculated according to eq. \ref{eq:poigen}. Then, it terminates. 

As mentioned above, to reduce the complexity of the analysis, the ECC algorithm can generate only outcomes with probability greater than, or equal to, a certain threshold $\theta\ (0\le \theta<1)$. In this case, the algorithm stops to analyze a certain sequence of events as soon as it discovers that its probability to occur is lower than $\theta$. For instance, let us assume that $\mathbb{P}\{e_1\}=0.95$ while $\mathbb{P}\{e_2\}=0.05$. In case $\theta=0.1$, the algorithm analyzes only the sequences composed of events highlighted in red in figure \ref{fig:graphoutcomes}, i.e. the outcomes starting with event $e_1$. This is because, since $\mathbb{P}\{e_2\}=0.05$, all the outcomes starting with event $e_2$ will have a probability to occur lower than or equal to $0.05<\theta$. 

The steps performed by the ECC algorithm are reported by the flowchart in Figure \ref{fig:flowchart-model}. Before describing it we define the concept of \emph{event} and \emph{chain of events}. 

\newtheorem{thm}{Definition}
\begin{thm}
An \textbf{event} $e_i=\{T_i,\ t_i\}$ represents a transmission occurring in the network, where $T_i$ indicates the {\em type} of the event and $t_i$ denotes its {\em starting time}. The event type can be either a success ($T_i=S$) or a failure ($T_i=F$). A success occurs whenever a node successfully transmits its packet, while a failure happens when two or more nodes transmit their packets simultaneously and, therefore, a collision occurs. The {\em starting time} $t_i$ of an event $e_i$ is defined as the time instant at which the node(s) causing $e_i$ start their CCA. Each event $e_i$ is also associated with a {\em finish time} $f_i$, defined as the first time instant, following $e_i$, at which a new event can occur, i.e. as the time $t^* > t_i$ such that a (new) successful CCA can be performed. \label{def:event} 
\end{thm}

\begin{figure}[htbp]
	\centering
		\includegraphics[scale=0.55]{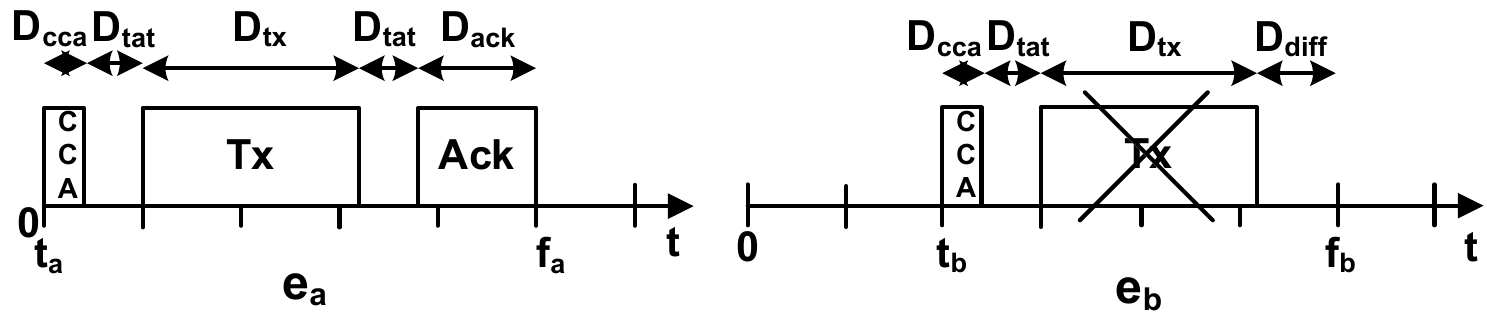}
	\caption{Events $e_a$ and $e_b$.}
	\label{fig:Figure1}
\end{figure}

Figure \ref{fig:Figure1} depicts two possible events, i.e., a successful transmission $e_a$ and a transmission failure $e_b$ (time is divided in slots of duration equal to the backoff period, $D_{bp}$). In Figure \ref{fig:Figure1}, $D_{cca}$ is the duration of CCA, $D_{tat}$ is the turnaround time (i.e., the time needed to switch the radio from receive to transmission mode, or vice versa), $D_{ack}$ is the time to transmit/receive an acknowledgement and $D_{diff}= D_{bp}- (D_{tx}\bmod{D_{bp}})$ is the time between the end of a packet transmission and the beginning of the next slot. In this example, the starting time of event $e_a$ is $0$, while the starting time of event $e_b$ is $t_b=2\cdot D_{bp}$. Hence, we can refer to $e_a$ as $\{S,\ 0\}$ and to $e_b$ as $\{F,\ 2 D_{bp}\}$.

\begin{thm} A \textbf{chain of events} $c=\{s_c,\ p_c,\ en_c\}$ represents a sequence of events that occur subsequently in the network. It is characterized by:
\begin{itemize}
\item the event sequence $s_c =\{e_1, ..., e_{m}\}$, where $m$ denotes the total number of events occurred;
\item the aggregate probability $p_{c}= \mathbb{P}\{e_1\wedge e_2\ ...\ \wedge e_m\}=\mathbb{P}\{e_1\}\mathbb{P}\{e_2\ |\ e_1\}\ ...\ \mathbb{P}\{e_m\ |\ e_1 \wedge e_2\ ...\wedge e_{m-1} \}$ that the event sequence $s_c$ occurs;
\item the average total energy $en_{c}$ spent by all the nodes in the network during the time interval [$0$, $f_{m}$], where $f_{m}$ is the finish time of the last event in sequence $s_c$. \qed
\end{itemize}
Please note that a chain $c:\ s_c =\{e_1, ..., e_{m}\}$ represents a possible \emph{outcome} of the CSMA/CA execution iff $\mathbb{P}\{no\_txs\ |\ e_1\wedge ...\wedge e_m \}>0$.
\label{def:chain} 
\end{thm}
\begin{figure}[!hbp]
\centering
\scalebox{0.6}{
\begin{tikzpicture}[node distance = 5cm, auto]
\tikzstyle{decision} = [diamond, draw, fill=blue!15, text badly centered, node distance=3cm, inner sep=3pt]
\tikzstyle{block} = [rectangle, draw, fill=blue!15, rounded corners, minimum height=4em]
\tikzstyle{block1} = [rectangle, draw, fill=green!15, rounded corners, minimum height=4em]
\tikzstyle{line} = [draw, -latex']
\tikzstyle{hub} = [circle, draw, thick]
    \node [block1, text width = 12em, text centered] (start) {START};
    \node [block, below of=start, text width= 10em, text centered, node distance = 2cm] (iniset) {$L_{c}=\{\emptyset\}$, $F_{c}=\{\emptyset\}$};
    \node [block, below of=iniset, text width = 11em, text centered, node distance = 2cm] (1iniset) {Initialization of $L_c$ with the initial chains};
    \node [decision, below of=1iniset, align=center] (decend){is\\ $L_{c}=\{\emptyset\}$\\?};
    \node [block, right of=decend, node distance = 5cm, align=center] (metrics){Derivation of performance\\ metrics using set $F_c$};
    \node [block1, below of=metrics, node distance=2cm, text width=6em, text centered] (end){END};
    \node [block, below of=decend, node distance=3cm, align=center] (extract) {Extract event chain\\ $c:\ s_c=\{e_1,\ e_2,\ ...,\ e_m\}$ \\from $L_c$};
    \node [decision, below of=extract, node distance=3cm, align=center] (decfin) {$\mathbb{P}\{no\_txs\ |\ c\}$\\ $>0$?};
    \node [hub, below of=decfin, node distance=2 cm] (hub) {};
    \node [decision, below of=hub, node distance=2.5cm, align=center] (oe) {$\mathbb{P}\{no\_txs\ |\ c\}$\\ $\neq 1$?};
    \node [hub, left of=oe, node distance=5cm] (hub2) {};
    \node [block, below of=oe, node distance=3cm, align=center] (theta) {Derive all the possible events $e_x$ that can occur after \\ $c$, and their probability $\mathbb{P}\{e_x\ |\ c\}$. For each $e_x$, add \\$c_x:\ s_{c_x}=\{e_1,\ e_2,\ ...,\ e_m,\ e_x\}$  to $L_c$ iff $p_{c_x} \ge \theta$.};
    \node [block, right of=decfin, node distance = 5cm, align=left] (addend){1) $c_x=c$.\\ 2) $p_{c_x}=p_c\cdot \mathbb{P}\{no\_txs\ |\ c\}$.\\ 3) If $p_{c_x}\ge \theta$, calculate $en_{c_x}$ \\ and add $c_x$ to $F_c$ since \\ it represents an outcome};
    \path [line] (start) -- (iniset);
    \path [line] (iniset) -- (1iniset);
    \path [line] (1iniset) -- (decend);
    \path [line] (decend) -- (metrics);
    \path [line] (metrics) -- (end);
    \path [line] (decend) -- (extract);
    \path [line] (extract) -- (decfin);
    \path [line] (decfin) -- (hub);
    \path [line] (decfin) -- (addend);
    \path [line] (hub) -- (oe);
    \path [line] (oe) -- (theta);
    \path [line] (oe) -- (hub2);
    \path [line] (hub2) -- (-5,-7) -- (decend);
    \path [line] (addend) -- (5, -15) -- (hub);
    \path [line] (theta) -- (-5, -20.5) -- (hub2);
    \coordinate [label=left:NO] (n) at (1, -8.8);
    \coordinate [label=left:YES] (y) at (2.5, -7.4);
    \coordinate [label=left:YES] (y2) at (2.5, -13.3);
    \coordinate [label=left:NO] (n2) at (1, -14.6);
    \coordinate [label=left:YES] (y3) at (1, -19.2);
    \coordinate [label=left:NO] (n3) at (-2, -18);
\end{tikzpicture}
}
\caption{Steps performed by the ECC algorithm.}
\label{fig:flowchart-model}
\end{figure}
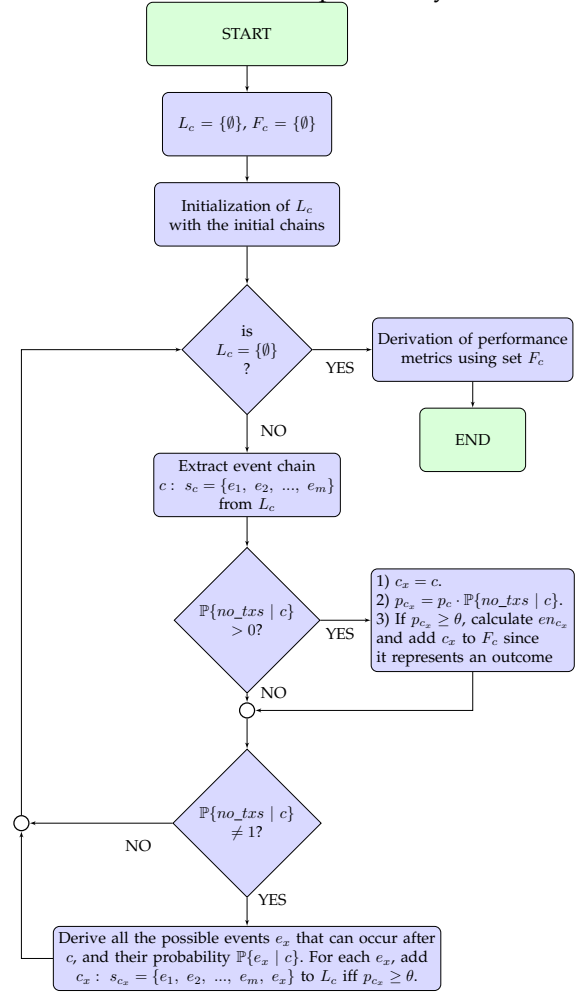

Hereafter, for brevity, we indicate as $\mathbb{P}\{e_x\ |\ c\}=\mathbb{P}\{e_x\ |\ e_1\wedge ...\wedge e_m \}$ the probability that event $e_x$ occurs in the network, given that the sequence of events $s_c =\{e_1, ..., e_{m}\}$ has occurred. Also, we denote by $\mathbb{P}\{no\_txs\ |\ c\}=\mathbb{P}\{no\_txs\ |\ e_1\wedge ...\wedge e_m \}$ the probability that no other events will occur in the network after the sequence represented by chain $c$.

To derive all the outcomes of the CSMA/CA execution, and the related probability, ECC follows an iterative approach summarized in Figure \ref{fig:flowchart-model}. Initially, ECC creates two empty sets, namely $L_c$ and $F_c$. At a given point in time, $L_c$ contains the chains still to be analyzed by the algorithm, while $F_c$ contains chains representing possible outcomes of the CSMA/CA execution. ECC starts analyzing the network at time $t=0$ and derives all the possible events $e_i$ that can occur just after $t=0$. For each such event $e_i$, the chain $c: s_c=\{e_i\}$ is added to set $L_c$, iff $p_c=\mathbb{P}\{e_i\}\ge \theta$. Then, the ECC algorithm enters a loop that ends when there are no more chains to be analyzed, i.e., $L_c=\{\emptyset\}$. At each iteration, a chain $c: s_c=\{e_1,\ e_2, ...,\ e_m\}$ is extracted from set $L_c$ to be analyzed. First, the algorithm checks if $c$ can be a possible outcome of the CSMA/CA execution, i.e. if $\mathbb{P}\{no\_txs\ |\ c\}>0$. If so, the following operations are performed. First, a copy $c_x$ of chain $c$ is created. Second, since no more events have to occur in the network to consider $c_x$ an outcome, the probability of chain $c_x$ is updated as $p_{c_x}=p_c\cdot \mathbb{P}\{no\_txs\ |\ c\}$. Finally, if $p_{c_x}\ge \theta$, the average energy $en_{c_x}$ spent by the nodes in the network when the events reported by $c_x$ occur is calculated and $c_x$ is added to $F_c$ since it represents an outcome of the CSMA/CA execution with a probability to occur $\ge \theta$.

If $\mathbb{P}\{no\_txs\ |\ c\} \neq 1$, it means that other events can occur in the network after those in $c$. In this case, the algorithm derives all the events $e_x$ (and their probability $\mathbb{P}\{e_x\ | c\}$) that can occur {\em after} the last event $e_m$ in chain $c$. For each such event $e_x$, the chain $c_x: s_{c_x}=\{e_1,\ e_2, ...,\ e_m,\ e_x\}$ is added to set $L_c$, provided that the corresponding probability $p_{c_x}=\mathbb{P}\{e_x\ | c\}\cdot p_c$ is greater than, or equal to, $\theta$. 
When the set $L_c$ becomes empty it means that ECC has generated all the possible outcomes having a probability to occur greater than, or equal to, $\theta$. Hence, the ECC algorithm proceeds with deriving the performance metrics of interest, using the chains in set $F_c$. Then, it terminates its execution. 
\section{Model derivation}
Now we detail each single step of the ECC algorithm. After a preliminary analysis in section \ref{subs:basic}, in Section \ref{subs:ecci} we focus on the ECC initialization phase and derive all the possible events that can occur in the network just after time $t=0$ and the corresponding probabilities. In Section \ref{subs:ce}, we focus on the actions performed by ECC inside the loop (chains examination phase). Finally, in Section \ref{subs:parallel} we show how to parallelize ECC while in Section \ref{subs:pm} we derive performance metrics of interest.

\subsection{Preliminaries}\label{subs:basic}
Before proceedings into the details of the ECC algorithm, we derive a general formula for the probability that a sensor node performs a CCA (Clear Channel Assessment) at a given time $t$. To this end, we first derive the possible time instants at which a sensor node could start a CCA. Next, we compute the probability that it actually performs a CCA in one of these time instants.

As a preliminary step, we need to consider all the actions that may lead a node to start a CCA at a given time $t$. Let $B_{max}=\mbox{{\em macMaxCSMABackoffs}}+1$ denote the maximum number of consecutive CCAs allowed for each transmission attempt, and $T_{max}=\mbox{{\em macMaxFrameRetries}}+1$ be the maximum number of transmission attempts allowed per data packet. In addition, let us indicate by $W_{i}$, $1\le i \le B_{max}$, the backoff window size at the \emph{i}-th backoff stage. For simplicity, hereafter we will use the expression "the sensor node is in state $B_{ij}$" to indicate that a sensor node is performing a CCA  during the \emph{i}-th backoff stage of the \emph{j}-th transmission attempt. Now we derive the set $\Lambda_{ij}$ of all the possible instants at which a sensor node could start a CCA while in state $B_{ij}$. 
\label{def:bij}
\begin{figure}[htbp]
	\centering
		\includegraphics[scale=0.58]{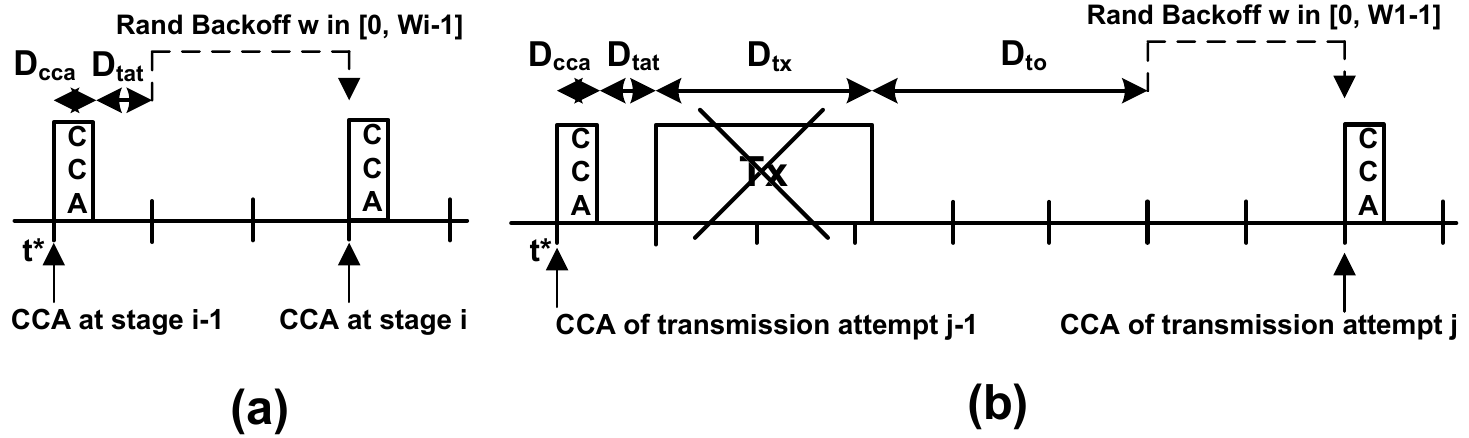}
		\caption{CCA due to a failed CCA (a) and a failed transmission (b).}
	\label{fig:Figure45}
\end{figure}

According to the CSMA/CA algorithm, at $t = 0$, each sensor node waits for a random number $w \in [0,\ W_1-1]$ of backoff periods and, then, it performs a CCA. Hence, $\Lambda_{11} = \{0, D_{bp}, \cdots, (W_1 -1) \cdot D_{bp}\}$. Then, a sensor node can start a CCA after one of the two following events: \textbf{(i)} a previous CCA during which the channel was found busy (see Figure \ref{fig:Figure45}a), or \textbf{(ii)} an unsuccessful transmission attempt (see Figure \ref{fig:Figure45}b). In case (i), the sets $\Lambda_{ij},\ 2\le i \le B_{max}, 1 \le j \le T_{max}$, can be recursively derived from $\Lambda_{i-1j}$ as follows:
\begin{equation}
\begin{split}
\Lambda_{ij} = & \{t\ |\ \exists w \in [0, W_{i}-1] \wedge \exists t^* \in \Lambda_{i-1j}\\
		        & \wedge t = t^* + D_{cca} + D_{tat} + w \cdot D_{bp} \}
\end{split}
\label{eq1}
\end{equation}
Equation \ref{eq1} derives CCA instants in set $\Lambda_{ij}$ by considering all the CCA instants $t^*\in\Lambda_{i-1j}$ and all the possible random backoff values $w\in[0, W_i-1]$ a node can generate when it is in the \emph{i}-th backoff stage and has found the channel busy.

In case \textbf{(ii)} the sensor node performs a CCA due to a previous unsuccessful transmission attempt and, hence, it is in one of the states $B_{1j}, 2\le j \le T_{max}$. Let us denote by $D_{rtx} \triangleq D_{cca} + D_{tat} + D_{tx} + D_{to}$ the total time needed to perform a CCA ($D_{cca}$), turn the radio in TX mode ($D_{tat}$), transmit a data packet ($D_{tx}$), and wait for the timeout ($D_{to}$). We indicate as  $R_{ij-1}$ the set of instants at which a node could perform a CCA after an unsuccessful transmission started at any $t^*\in \Lambda_{ij-1}, 1\le i \le B_{max}$. It can be expressed as follows.
\begin{equation}
\begin{split}
R_{ij-1} = &\{t\ |\ \exists w \in [0, W_{1} - 1] \wedge \exists t^* \in \Lambda_{ij-1} \\
		& \wedge t = t^* + D_{rtx} +  w \cdot D_{bp}\}
\end{split}
\label{eq:rij}
\end{equation}
Equation \ref{eq:rij} calculates  $R_{ij-1}$ considering all time instants $t^* \in \Lambda_{ij-1}$ and all possible backoff values $w\in[0, W_1-1]$ the node can generate at the first backoff stage. Since a retransmission can occur during every backoff stage, the set $\Lambda_{1j}, 2\le j \le T_{max}$, is computed as the union of all the sets $R_{ij-1}$,  i.e., $\Lambda_{1j} = \bigcup_{i = 1}^{B_{max}} R_{ij-1}$.

Let $\Omega_{ij}^{t}$ denote the set of all time instants $t^*$ at which a node can perform a CCA before performing a CCA at time $t$ during state $B_{ij}$. The following claim holds. 
\newtheorem{lem}{Claim}
\begin{lem}
The set $\Omega_{ij}^{t}$ can be derived as
\begin{equation}
\Omega_{ij}^{t} = \left\{
\begin{array}{ll}
\{\emptyset\},\  \ \ \ \ \ \ \mbox{if}\ i = 1,\ j = 1
\\
\\
\{t^* \in \Lambda_{i-1j}\ |\ \exists w \in [0, W_{i}-1]\wedge\\ \ \  \ \ \ \ \ \ \ \ \ \  t = t^* + D_{cca} + D_{tat} + w \cdot D_{bp}\},\\ \\ \ \ \ \ \ \ \ \ \ \ \ \ \mbox{if} \ 2 \le i \le B_{max},\ 1 \le j \le T_{max}\\
\\
\{t^* \in \Lambda_{i'j-1},\ 1\le i' \le B_{max}\ |\ \\\exists w \in [0, W_{1}-1]: t= t^* +  D_{rtx} + w \cdot D_{bp}\},\\ \ \ \ \ \ \ \ \ \ \ \ \ \mbox{if}\ i = 1,\ 2 \le j \le T_{max}
\end{array}
\right.
\label{eq:pred}
\end{equation}
\end{lem}
\begin{proof} 
The set $\Omega_{11}^{t}$ is empty since no CCA can be performed before those occurring at time instants in set $\Lambda_{11}$. If $2 \le i \le B_{max}$, it means that the node performs a CCA at time $t$ due to a previous failed CCA. In this case, all the CCA instants $t^* \in \Lambda_{i-1j}$ that can result in a CCA at $t$ are selected (second term of Equation \ref{eq:pred}). In the last case, a CCA at time $t$ is due to a previous unsuccessful transmission. Therefore, the set $\Omega_{1j}^{t}, 2 \le j \le T_{max}$, is composed by all the $t^* \in \Lambda_{i'j-1}, 1\le i'\le B_{max}$, which could cause the node to perform a CCA at $t$ due to an unsuccessful transmission (third term of Equation \ref{eq:pred}). 
\end{proof}

Let us now derive the probability $\mathbb{P}\{\mbox{CCA}^{t}\}$ that a sensor node performs a CCA at time $t$. To this end, we calculate the probability $\mathbb{P}\{\mbox{CCA}_{ij}^{t}\}$ that a node performs a CCA at time $t$ while in state $B_{ij}$ and, then, we compute $\mathbb{P}\{\mbox{CCA}^{t}\}$ based on $\mathbb{P}\{\mbox{CCA}_{ij}^{t}\}$. Let us denote by $\mathbb{P}\{\mbox{CB}^{t}\}$ the probability to find the channel busy during a CCA started at time $t$, and by  $\mathbb{P}\{\mbox{F}^{t}\}$ the probability that a transmission whose CCA started at time $t$ fails. The following claims hold.

\begin{lem}
The probability $\mathbb{P}\{\mbox{CCA}_{ij}^{t}\}$ that a sensor node performs a CCA at $t$ while in state $B_{ij}$ is
\begin{equation}
\mathbb{P}\{\mbox{CCA}_{ij}^{t}\} = \left\{
\begin{array}{ll}
0 \ \ \ \ \ \ \ \ \ \ \mbox{if}\ t \not \in \Lambda_{ij} 
\\
\frac{1}{W_1},\ \ \ \ \ \   \mbox{if}\ i = 1,\ j = 1 
\\
\\
\sum_{t^* \in \Omega_{ij}^{t}} \mathbb{P}\{\mbox{CCA}_{i-1j}^{t^*}\} \cdot \mathbb{P}\{\mbox{CB}^{t^*}\} \cdot \frac{1}{W_{i}},\\ \ \ \ \ \ \ \ \ \ \  \mbox{if}\ 2 \le i \le B_{max}, 1 \le j \le T_{max}
\\
\\
\sum_{t^* \in \Omega_{ij}^{t}} \sum_{i'=1}^{B_{max}} \mathbb{P}\{\mbox{CCA}_{i'j-1}^{t^*}\} \cdot \\  \qquad \ \ \ \ \ \ \cdot (1 - \mathbb{P}\{\mbox{CB}^{t^*}\}) \cdot \mathbb{P}\{\mbox{F}^{t^*}\} \cdot \frac{1}{W_{1}}, \\ \ \ \ \ \ \ \ \ \ \ \ \mbox{if}\ i = 1,\  2 \le j \le T_{max}
\end{array}
\right.
\label{eq:pcs}
\end{equation}
\end{lem}

\begin{proof} 
See Appendix A.
\end{proof}

\begin{lem} 
The probability that a sensor node performs a CCA at a certain time $t$ can be calculated as:
\begin{equation}
\mathbb{P}\{\mbox{CCA}^{t}\} = \sum_{i=1}^{B_{max}} \sum_{j=1}^{T_{max}} \mathbb{P}\{\mbox{CCA}_{ij}^{t}\}
\label{eq:pcsg}
\end{equation}
\end{lem}
\begin{proof}
$\mathbb{P}\{\mbox{CCA}^{t}\}$ is equal to the probability that a sensor node performs a CCA at time $t$ in any state $B_{ij}$. Therefore, Equation \ref{eq:pcsg} calculates $\mathbb{P}\{\mbox{CCA}^{t}\}$ as the sum of $\mathbb{P}\{\mbox{CCA}_{ij}^{t}\}, 1\le i \le B_{max}, 1 \le j \le T_{max}$. Since events "performing a CCA at time $t$ while in state $B_{ab}$" and "performing a CCA at time $t$ while in state $B_{cd}$", $ a \neq c\ |\ b \neq d$, are always mutually exclusive, it is possible to sum probabilities $\mathbb{P}\{\mbox{CCA}_{ij}^{t}\}$. 
\end{proof}

\subsection{ECC Initialization}\label{subs:ecci}
As shown in Figure \ref{fig:flowchart-model}, the first step of the ECC algorithm consists in initializing the set $L_c$ with chains derived from events occurring immediately after $t = 0$. In this section, we will refer to $e_{s_i}$ ($e_{f_i}$) as the success (failure) event starting at time $i \cdot D_{bp},\ i\in \mathbb{N}$. Also, we denote by $\mathbb{P}\{e_{s_i}\}$ ($\mathbb{P}\{e_{f_i}\})$ the probability that event $e_{s_i}$ ($e_{f_i}$) occurs.

According to the CSMA/CA algorithm, at $t = 0$ each sensor node waits for a random number $w\in [0,\ W_1-1]$ of backoff periods and, then, it performs a CCA. Therefore, the first event occurring in the network can be either a success or a failure with starting time in the set $\{0,\ D_{bp},\ 2D_{bp},\ldots,\ (W_1-1)\cdot D_{bp}\}$. 


A successful transmission occurs at time $i\cdot D_{bp}\ (i=0, .., W_1-1)$ when one node generates a backoff time equal to $i\cdot D_{bp}$, and all the other $N-1$ nodes extract a backoff time larger than $i\cdot D_{bp}$. Therefore, 
\begin{equation}
\mathbb{P}\{e_{s_i}\} = N \cdot \frac{1}{W_1} \cdot \left(\frac{W_1-i-1}{W_1}\right)^{N-1}
\label{eq:esi0}
\end{equation}
In Equation (\ref{eq:esi0}), the term ${1/W_1}$ is the probability that one node picks up a backoff time equal to $i\cdot D_{bp}$, while the third term gives the probability that all the remaining $N-1$ nodes generate a backoff time larger than $i\cdot D_{bp}$. Conversely, a failure occurs at time $i\cdot D_{bp}$ when two or more nodes generate the same backoff time $i\cdot D_{bp}$ and, thus, experience a collision. Hence, 
\begin{equation}
\mathbb{P}\{e_{f_i}\} = \sum\limits_{k=2}^{N} \binom{N}{k} \left(\frac{1}{W_1}\right)^k \cdot \left(\frac{W_1-i-1}{W_1}\right)^{N-k}
\label{eq:efi0}
\end{equation}
The sum in Equation (\ref{eq:efi0}) takes into account that more than two nodes may collide. The term inside the sum gives the probability that exactly $k$ nodes randomly pick up a backoff time of $i\cdot D_{bp}$, while $N-k$ nodes choose a backoff value larger than $i\cdot D_{bp}$. 

Using Equations (\ref{eq:esi0}) and (\ref{eq:efi0}), ECC initializes $L_c$ by adding chains $c$: $s_c = \{e_{s_i}\}$ ($s_c=\{e_{f_i}\}$) and $p_c=\mathbb{P}\{e_{s_i}\}$ ($p_c=\mathbb{P}\{e_{f_i}\}$). Note that a chain is added to $L_c$ iff $p_c\ge \theta$. Then, ECC enters the \emph{chains examination} phase.

\subsection{Chains examination}\label{subs:ce}
In the chains examination phase, ECC executes a loop during which, at each step, a chain $c \in L_c$ with $s_c:\ \{e_1,\ ...,\ e_m\}$ is examined. The goal of the examination is twofold. First, the algorithm checks if $c$ represents a possible outcome of the CSMA/CA execution by computing $\mathbb{P}\{no\_txs\ |\ c\}$ and, if so (i.e. $\mathbb{P}\{no\_txs\ |\ c\}>0$), it adds $c$ to $F_c$. If $\mathbb{P}\{no\_txs\ |\ c\}\neq 1$ it means that new events may occur after $c$. Hence, as a second step, all the events that may occur after $e_m$ (i.e. after the last event in $c$) are derived. Hereafter, for simplicity, we will indicate as $e_{s_i}\ (e_{f_i})$ a success (failure) event occurring at time $t_i=f_m+i\cdot D_{bp},\ i\in \mathbb{N}$, where $f_m$ is the finish time of event $e_m$ and as $\mathbb{P}\{e_{s_i}\ |\ c\}\ (\mathbb{P}\{e_{f_i}\ |\ c\})$ its corresponding probability. For any $e_{s_i}\ (e_{f_i})$ a new chain $c_x: s_{c_x}=\{e_1,\ ...,\ e_m,\ e_{s_i}\}\ (s_{c_x}=\{e_1,\ ...,\ e_m,\ e_{f_i}\})$ is added to $L_c$ by the ECC algorithm iff $p_{c_x}=p_c\cdot \mathbb{P}\{e_{s_i}\ |\ c\}(p_c\cdot \mathbb{P}\{e_{f_i}\ |\ c\})\ge \theta$.

In the following we will show the computation of both $\mathbb{P}\{no\_txs\ |\ c\}$ and $\mathbb{P}\{e_{s_i}\ |\ c\}\ (\mathbb{P}\{e_{f_i}\ |\ c\})$. However, before deriving them, we perform two preliminary steps. First, we derive $\mathbb{P}\{CCA^t\ |\ c\}$, i.e. the probability that a node, that has not experienced a success until $f_m$, will perform (has performed) a CCA at a certain time $t\ge f_m$ ($t<f_m$), given that the events in $c$ occurred. Second, we compute the exact number of nodes that are still active in the network at time $t=f_m$ (i.e. that have not yet terminated the CSMA/CA execution). 
\subsubsection*{\textbf{Derivation of $\mathbb{P}\{CCA^t\ |\ c\}$}}
Hereafter, we take into consideration a generic chain $c:\ s_c=\{e_1,\ ...,\ e_m\}$ and denote by $N_s$ the number of successful transmissions occurred in chain $c$ ($N_s=\left|\{e_i\in s_c:\ T_i=S\}\right|$). Since each sensor node has to transmit just one data packet, at most $N_r = N - N_s$ nodes may be still active in the network after time $f_m$. Our goal is to derive the probability $\mathbb{P}\{CCA^t\ |\ c\}$ that any of the $N_r$ nodes will perform (has performed) a CCA at time $t$, given that the sequence of events in $c$ occurred. 

First of all, we denote by $N_{P_{ij}}, 1\le i\le B_{max}, 1\le j \le T_{max}$, the set of all time instants $t < f_{m}$ at which it is not possible, for any of the $N_r$ sensor nodes, to have performed a CCA while in state $B_{ij}$, if the events in $c$ occurred. The derivation of $N_{P_{ij}}$ is shown in Appendix B, due to the sake of space. 

Second, we derive, for time instants $t < f_{m}$\footnote{We assume that both $\mathbb{P}\{CB^t\ |\ c\}$ and $\mathbb{P}\{F^t\ |\ c\}$ are equal to $0$ $\forall t \ge f_m$. This allows to calculate the probability that a node will \emph{directly} perform a CCA at a time $t \ge f_m$.}, the probability $\mathbb{P}\{CB^t\ |\ c\}$ for any of the $N_r$ to have found the channel busy during a CCA started at time $t$. We also derive $\mathbb{P}\{F^t\ |\ c\}$, i.e. the probability for any of the same nodes to have experienced a failure for a transmission whose CCA started at time $t$. Equations \ref{eq:pcbc} and \ref{eq:pfailc} hold.
\begin{equation}
\mathbb{P}\{CB^t\ |\ c\} = \left\{
\begin{array}{ll}
1 &  \mbox{if }  \exists e_i \in s_c: t\in[t_i+D_{bp},\ f_i) \\
0 & \mbox{otherwise}
\end{array}
\right.
\label{eq:pcbc} 
\end{equation}
\begin{equation}
\mathbb{P}\{F^t\ |\ c\} = \left\{
\begin{array}{ll}
1 &  \mbox{if } \exists e_i \in s_c: t_i=t \wedge T_i=F \\ 
0 &  \mbox{otherwise}
\end{array}
\right.
\label{eq:pfailc} 
\end{equation}
The probability that a generic node has found the channel busy during a CCA started at time $t$ only depends on the specific events in chain $c$. Specifically, $\mathbb{P}\{CB^t\ |\ c\}=1$ if a success or failure event has occurred at time $t$, and zero otherwise. Similarly, $\mathbb{P}\{F^t\ |\ c\}= 1 $ if a failure occurred at time $t_i = t$, and zero otherwise. 

Finally, we indicate as $S_{max}$ ($S_{max}>f_m$), the largest time instant at which any of the $N_r$ nodes can perform a CCA, given that all the events in $c$ occurred. $S_{max}$ represents the largest instant at which a new event can occur after $e_m$. The following claim holds.

\begin{lem} 
$S_{max}=f_{m} + M_w \cdot D_{bp}$, where 
\begin{equation}
M_w = \left\{
\begin{array}{ll}
W_{B_{max}} - 1 & T_{m}=S \\ 
\mbox{max}\{W_{B_{max}}-1,\ 2+ (W_1 - 1)\}&  T_{m}=F
\end{array}
\right.
\label{eq:mw}
\end{equation}
\end{lem}
\begin{proof}
If $T_{m} = S$, $S_{max}$ is derived by considering the worst case shown in fig. 6a where a node (node A) performs a CCA at time $t=f_{m}-D_{bp}$ during its $(B_{max}-1)$-th backoff stage, finds the channel busy, and extracts a value for the backoff time equal to $(W_{B_{max}}-1)\cdot D_{bp}$. Hence, $S_{max}$ is given by $f_{m} + (W_{B_{max}}-1) \cdot D_{bp}$. Conversely, if $T_{m}=F$, two cases must be considered. A node (node A in fig. 6b) may perform a CCA at time $t=f_{m}-D_{bp}$ during the $(B_{max}-1)$-th backoff stage, and extract a backoff time value equal to $(W_{B_{max}}-1)\cdot D_{bp}$. At the same time, another node (node B in fig. 6b), which has experienced the collision represented by $e_{m}$, waits for the retransmission timeout $D_{to}$ and extracts a backoff time equal to $(W_1-1) \cdot D_{bp}$. Hence, we need to consider the largest value for $S_{max}$ in the two cases, i.e. $S_{max}=f_{m} + max((W_{B_{max}}-1),\ 2+(W_1-1))\cdot D_{bp}$.
\end{proof}
Now, we can show the computation of $\mathbb{P}\{CCA^t\ |\ c\}$ for time instants $t\in [0,\ S_{max}]$. The following claim holds.
\begin{lem}
The probability $\mathbb{P}\{CCA_{ij}^t\ |\ c \}$ that any of the $N_r$ nodes has performed a CCA at time $t\in[0,\ f_m]$ or will perform a CCA at a time $t\in[f_m,\ S_{max}]$, while in state $B_{ij}$, provided that all the events in chain $c$ have occurred, is 
\begin{equation}
\begin{split}
\footnotesize
\left\{
\begin{array}{ll}
0 \ \ \ \ \ \ \ \ \ \ \ \ \ \ \ \ \ \ if\ t \not \in \Lambda_{ij} \vee t \in N_{P_{ij}}
\\
\\
\frac{1}{\arrowvert \Lambda_{11} \setminus N_{P_{11}} \arrowvert},\ \ \ \ \ \   if\ i = 1,\ j = 1
\\
\\
\sum_{t^* \in (\Omega_{ij}^{t}\setminus N_{P_{i-1j}})} \mathbb{P}\{\mbox{CCA}_{i-1j}^{t^*}\ |\ c\} \cdot \mathbb{P}\{\mbox{CB}^{t^*}\ |\ c\} \cdot \frac{1}{\arrowvert h_{ij}^{t^*}\arrowvert},
\\ \ \ \ \ \ \ \ \ \ \ \ \ \ \ \ \ \ \  if\ 1 < i \le B_{max}, 1 \le j \le T_{max}
\\
\\
\sum_{i'=1}^{B_{max}} \sum_{t^* \in (\Omega_{ij}^{t} \setminus N_{P_{i'j-1}})}  \mathbb{P}\{\mbox{CCA}_{i'j-1}^{t^*}\ |\ c \} \cdot (1 - \mathbb{P}\{\mbox{CB}^{t^*}\ |\ c\}) \cdot \\  \qquad \ \ \ \ \ \ \cdot \ \mathbb{P}\{\mbox{F}^{t^*}\ |\ c\} \cdot \frac{1}{\arrowvert h_{1j}^{t^*}\arrowvert},  \\ \ \ \ \ \ \ \ \ \ \ \ \ \ \ \ \ \ \ if\ i = 1,\  2 \le j \le T_{max}\\
\end{array}
\right.
\label{eq:pcsijc}
\end{split}
\end{equation}
\end{lem}
\begin{proof} 
See Appendix C.
\end{proof}

\begin{figure}[tbp]
\centering
   \includegraphics[scale=0.45]{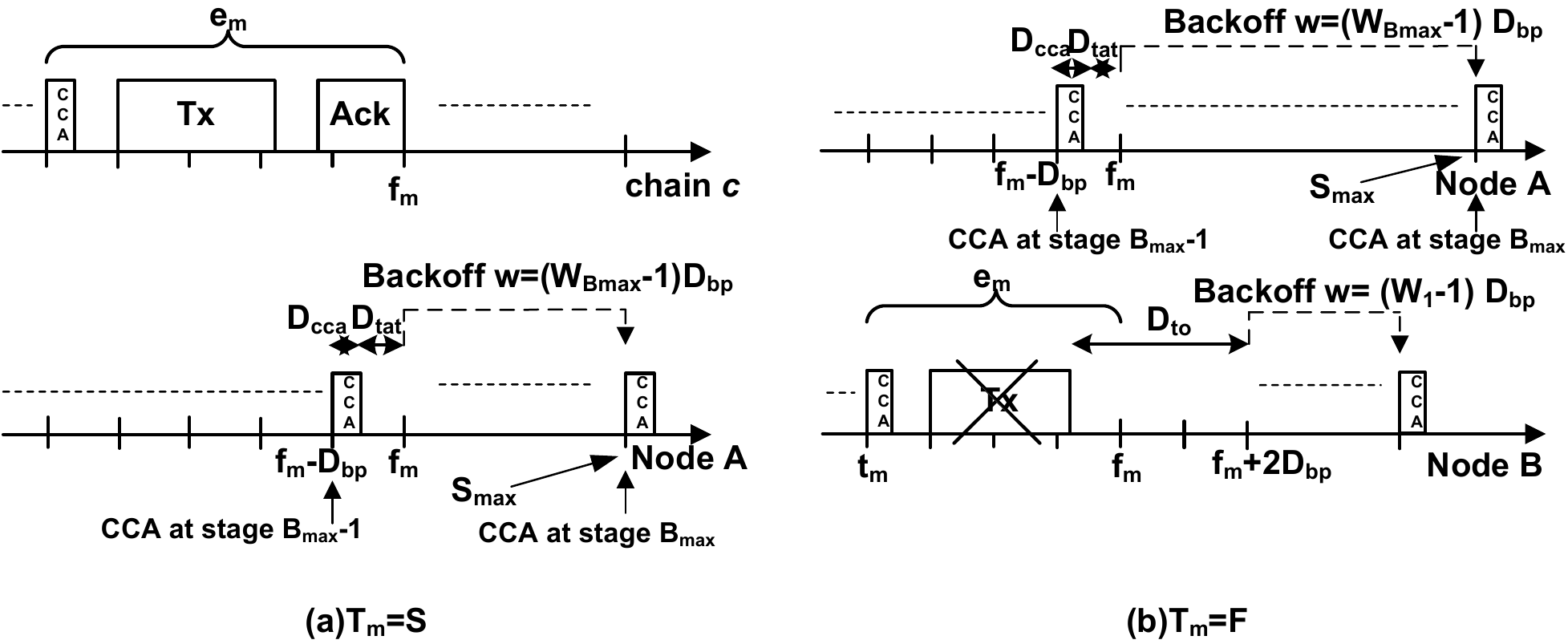}
 \caption{Derivation of $S_{max}$}
\end{figure}
     
Finally, we can derive $\mathbb{P}\{CCA^t\ |\ c\}$ as follows
\begin{equation}
\mathbb{P}\{\mbox{CCA}^{t}\ |\ c \} = \sum_{i=1}^{B_{max}} \sum_{j=1}^{T_{max}} \mathbb{P}\{\mbox{CCA}_{ij}^{t}\ |\ c \}
\label{eq:pcstc}
\end{equation}
\subsubsection*{\textbf{Estimating the number of active nodes at time $t=f_m$}}
Since $N_s$ nodes experienced a success during the chain $c$, at most $N_r=N-N_s$ nodes can be potentially active at time $t=f_m$. Each of these $N_r$ nodes can be in one of the following states at time $t=f_m$: \textbf{(i)} the node has reached the maximum number $B_{max}$ of consecutive CCAs for a data packet transmission; or \textbf{(ii)} the node has reached the maximum number $T_{max}$ of retransmissions for a data packet; \textbf{(iii)} the node is really active, i.e. it has not yet finished the CSMA/CA execution. Indeed, in the first two cases the sensor node drops its data packet, according to the CSMA/CA algorithm and, thus, it is no longer active at time $t=f_m$. To derive the number of sensor nodes that are really active at time $t=f_m$, we need to calculate the probability, for each of the $N_r$ nodes, to be in state \textbf{(i)}, \textbf{(ii)} or \textbf{(iii)}, respectively. The following claims hold.
\begin{lem}
Let $\mathbb{P}\{F_{CCA}\ |\ c \}$ denote the probability that any of the $N_r$ sensor nodes has exceeded the maximum number $B_{max}$ of consecutive CCAs allowed for the transmission of a data packet before time $t=f_m$. It is
\begin{equation}
\mathbb{P}\{F_{CCA}\ |\ c \} = \sum_{t=0}^{f_{m}-D_{bp}} \sum_{j=1}^{T_{max}} \mathbb{P}\{CCA_{B_{max}j}^t\ |\ c \} \cdot \mathbb{P}\{CB^t\ |\ c\}
\label{eq:pfcca}
\end{equation}
\end{lem}
\begin{proof}
See Appendix D.
\end{proof}
\begin{lem}
Let $\mathbb{P}\{F_{Rtx}\ |\ c\}$ denote the probability that any of the $N_r$ sensor nodes has exceeded the number of retransmissions allowed for a data packet before time $t=f_m$. It is
\begin{equation}
\mathbb{P}\{F_{Rtx}\ |\ c \} = \sum_{t=0}^{f_{m}-D_{bp}} \sum_{i=1}^{B_{max}} \mathbb{P}\{CCA_{iT_{max}}^t\ |\ c \} \cdot \mathbb{P}\{F^t\ |\ c\}
\label{eq:pfrtx}
\end{equation}
\end{lem}
\begin{proof}
See Appendix E.
\end{proof}
\begin{lem}
Let $\mathbb{P}\{A_{p}\ |\ c \}$ denote the probability that any of the $N_r$ nodes: (i) is still active at $t=f_m$ and (ii) it has participated to the last event $e_{m}$ in chain $c$, i.e. it has performed a CCA at time $t=t_{m}$. It is
\begin{equation}
\mathbb{P}\{A_{p}\ |\ c\} = \left\{
\begin{array}{ll}
0 & T_{m}=S \\ 
\sum_{i=1}^{B_{max}} \sum_{j=1}^{T_{max}-1} \mathbb{P}\{\mbox{CCA}_{ij}^{t_{m}}\ |\ c \}&  T_{m}=F
\end{array}
\right.
\label{eq:activend}
\end{equation}
\end{lem}
\begin{proof}
See Appendix F.  
\end{proof}
\begin{lem}
Let $\mathbb{P}\{A_{np}\ |\ c \}$ denote the probability that any of the $N_r$ sensor nodes: (i) is still active at time $t=f_m$, and (ii) it has not participated to the last event $e_{m}$ in chain $c$, i.e. it has not performed a CCA at time $t=t_m$. Hence,
\begin{equation}
\mathbb{P}\{A_{np}\ |\ c \} = \sum_{t=f_{m}}^{S_{max}} \mathbb{P}\{CCA^t\ |\ c\} -  \mathbb{P}\{A_{p}\ |\ c \}
\label{eq:activenoend} 
\end{equation}
\end{lem}
\begin{proof}
See Appendix G.
\end{proof}
Let us denote by $\overline{N}=[N_p,\ N_{np},\ N_d]$ a composition of sensor nodes, where \textbf{(i)} $N_p$ indicates the number of nodes that are still \emph{active} at $f_m$ and \emph{have} participated to event $e_{m}$ (i.e. have performed a CCA at time $t=t_m$), \textbf{(ii)} $N_{np}$ is the number of nodes that are still \emph{active} at $f_m$ but \emph{have not} participated to $e_{m}$ (i.e. have not performed a CCA at time $t=t_m$), and \textbf{(iii)} $N_d$ is the number of sensor nodes that have \emph{dropped} their packet and, hence, are no more active at $f_m$. By definition, $\forall \overline{N},\ \forall c,\ N_p + N_{np} + N_d = N_r$. We now compute the probability of $\overline{N}$, $\mathbb{P}\{\overline{N}\ |\ c\}$, both when $e_m$ is a success and when it is a failure.

If $e_m$ is a success $T_{m} = S$, then $\mathbb{P}\{A_{p}\ |\ c \}=0$. This is because it is not possible for a node to be still active after experiencing a success (each node has a single packet to transmit). Hence, $N_p=0,\ \forall \overline{N}$, and $\mathbb{P}\{\overline{N}\ |\ c\}$ is equal to:
\begin{equation}
\mathbb{P}\{\overline{N}\ |\ c\}=\binom{N_r}{N_{np}}\cdot \mathbb{P}\{A_{np}\ |\ c \}^{N_{np}} \cdot \mathbb{P}\{D\}^{N_d}
\label{eq:eq:case1nact}
\end{equation}

In Equation \ref{eq:eq:case1nact} $\mathbb{P}\{D\}=\mathbb{P}\{F_{CCA}\ |\ c \}+\mathbb{P}\{F_{Rtx}\ |\ c \}$ is the probability that a sensor node has dropped its data packet due to either exceeded number of backoff stages or exceeded number of retransmissions, before $f_m$. In addition, the second and third terms provide the probability that exactly $N_{np}$ nodes are still active in the network, and the probability that $N_d$ nodes are no more active, respectively. Obviously, all possible combinations are taken into consideration.

The calculation of $\overline{N}$ for the case $T_{m} = F$ follows the same line of reasoning and is shown in Appendix H.
\begin{lem}
The probability $\mathbb{P}\{no\_txs\ |\ c\}$ that no other events occur in the network after $e_m$ is:
\begin{equation}
\mathbb{P}\{no\_txs\ |\ c\}=\mathbb{P}\{\overline{N}=[0,\ 0,\ N_r]\ |\ c\}
\end{equation}
\end{lem}
\begin{proof}
The probability $\mathbb{P}\{no\_txs\ |\ c\}$ that no events occur after $e_m$, is equal to the probability that all nodes have finished their CSMA/CA execution before $f_m$, i.e. $\mathbb{P}\{\overline{N}=[0,\ 0,\ N_r]\ |\ c\}$.
\end{proof}
\subsubsection*{\textbf{Derivation of new events after $e_m$}}
When $\mathbb{P}\{no\_txs\ |\ c\}\neq 1$, it means that there are cases when at least one node is still active at time $f_m$ and, hence, other events may occur in the network. We consider all the events that may occur after $e_m$ and calculate the corresponding probability. To this end, we first derive the probability for an active sensor node to perform a CCA at a time $t\in[f_{m},\ S_{max}]$. We need to discriminate between active nodes that \emph{have} participated to event $e_{m}$, and active nodes that \emph{have not} participated.

In the former case, after participating to $e_{m}$, sensor nodes are in one of the states $B_{1j}, 2\le j\le T_{max}$, i.e. the first backoff stage of a retransmission attempt. Hence, according to CSMA/CA, they will perform a CCA after waiting for both the retransmission timeout $D_{to}$ and a random number $w\in [0,\ W_1-1]$ of backoff periods. Since $w$ is uniformly distributed in $[0,\ W_1-1]$, the probability $\mathbb{P}\{CCA_{p}^t\ |\ c\}$ that any of the considered sensor nodes will perform a CCA at time $t\in [f_m,\ S_{max}]$ is equal to $\frac{1}{W_1}$ for $t\in [f_{m}+2D_{bp},\ f_{m}+2D_{bp}+(W_1-1)D_{bp}]$, and zero otherwise (see also figure 6b).

In the second case, we consider nodes that have not participated to $e_{m}$. First, we correct the computation of $\mathbb{P}\{ CCA_{1j}^t\ |\ c \},\ 2\le j\le T_{max}$, for time instants $t\in [f_{m}+2\cdot D_{bp},\ f_{m}+2\cdot D_{bp}+(W_1-1)\cdot D_{bp}]$, in equation \ref{eq:pcsijc} as
\begin{equation}
\mathbb{P}\{CCA_{1j}^t\ |\ c \} = \mathbb{P}\{CCA_{1j}^t\ |\ c \} - \sum_{i=1}^{B_{max}} \mathbb{P}\{ CCA_{ij-1}^{t_{m}}\ |\ c\} \cdot \frac{1}{W_1}
\label{eq:corrpcsij}
\end{equation}
This is to exclude those cases that lead the sensor node to perform a CCA at time $t$ after performing a CCA at time $t_{m}$. Thus, $\mathbb{P}\{CCA_{np}^t\ |\ c\}$, for $t\in[f_{m},S_{max}]$, is:
\begin{equation}
\mathbb{P}\{CCA_{np}^t\ |\ c\}= \frac{\mathbb{P}\{CCA^t \ |\ c\}}{\mathbb{P}\{A_{np}\}}
\label{eq:pcsnoeend}
\end{equation}

Equation \ref{eq:pcsnoeend} can be explained as follows. Since we are considering an active sensor node, it will surely perform a CCA at a time $t\ge f_{m}$. Hence, for such a node, the probability to perform a CCA at a specific time $t\in[f_{m},\ S_{max}]$ can be calculated by normalizing $\mathbb{P}\{CCA^t \ |\ c\}$ with probability $\mathbb{P}\{A_{np}\}$, i.e., the probability that the node will perform a CCA at any instant $t\ge f_{m}$.

Now, we derive both $\mathbb{P}\{ e_{s_i}\ |\ \overline{N}\}$ and $\mathbb{P}\{ e_{f_i}\ |\ \overline{N}\}$, $i \in \{0,\ 1,\ ...,\ M_w\}$ where $\mathbb{P}\{ e_{s_i}\ |\ \overline{N}\}$ ($\mathbb{P}\{ e_{f_i}\ |\ \overline{N}\}$) denotes the probability that the first event that occurs after $e_{m}$, given the composition of nodes $\overline{N}$, is a success (failure) occuring at time $t_i=f_{m}+i\cdot D_{bp}$. Claim 11 holds.
\begin{lem} 
\begin{align*}
\mathbb{P}\{ e_{s_i}\ |\ \overline{N}\} = {} & N_{np} \cdot \mathbb{P}\{CCA_{np}^{f_{m}+i \cdot D_{bp}}\ |\ c\} \\
& \cdot \left(\sum_{j=i+1}^{M_w} \mathbb{P}\{CCA_{np}^{f_{m}+j \cdot D_{bp}}\ |\ c\}\right)^{N_{np}-1} \\
& \cdot \left(\sum_{j=i+1}^{M_w} \mathbb{P}\{CCA_p^{f_{m}+j \cdot D_{bp}}\ |\ c\}\right)^{N_{p}} \\
& + N_{p} \cdot \mathbb{P}\{CCA_p^{f_{m}+i \cdot D_{bp}}\ |\ c\} \cdot \\
& \cdot \left(\sum_{j=i+1}^{M_w} \mathbb{P}\{CCA_p^{f_{m}+j \cdot D_{bp}}\ |\ c\}\right)^{N_{p}-1}\\
&  \cdot \left(\sum_{j=i+1}^{M_w} \mathbb{P}\{CCA_{np}^{f_{m}+j \cdot D_{bp}}\ |\ c\}\right)^{N_{np}}
\end{align*}
\label{eq:pesic}
\end{lem}
\begin{proof} 
See Appendix I.
\end{proof}

Let us focus now on probability $\mathbb{P}\{ e_{f_i}\ |\ \overline{N}\}$.  Since a failure can occur only if $N_{np} + N_p \ge 2$ it follows that $\mathbb{P}\{ e_{f_i}\ |\ \overline{N}\}=0$, $\forall i \in [0,\ M_w]$ if $N_p + N_{np} < 2$. 
Below, we will focus on cases where $N_p + N_{np} \ge 2$. We denote by $comp(m)$ the set of compositions $(m_1,\ m_2)$ of an integer $m$ in two distinct parts $m_1$ and $m_2$ with $m_1+m_2=m$. 
\begin{lem}
\begin{align*}
\mathbb{P}\{ e_{f_i}\ |\ \overline{N}\}& = {} \sum_{m=2}^{N_p + N_{np}} \sum_{(m_1,\ m_2)  \in comp(m): m_1 \le N_p \wedge m_2 \le N_{np}}^{} \\
& \binom{N_p}{m_1}\cdot  \mathbb{P}\{CCA_p^{f_{m}+i \cdot D_{bp}}\ |\ c\}^{m_1}\\
& \left(\sum_{j=i+1}^{M_w}\mathbb{P}\{CCA_{p}^{f_{m}+j \cdot D_{bp}}\ |\ c\}\right)^{N_p-m_1} \\
& \cdot  \binom{N_{np}}{m_2} \mathbb{P}\{CCA_{np}^{f_{m}+i \cdot D_{bp}}\ |\ c\}^{m_2}\\
& \left(\sum_{j=i+1}^{M_w}\mathbb{P}\{CCA_{np}^{f_{m}+j \cdot D_{bp}}\ |\ c\}\right)^{N_{np}-m_2}
\end{align*}
\label{eq:pefic}
\end{lem}
\begin{proof} 
See Appendix J.
\end{proof}

Finally, using the law of total probability, we derive both $\mathbb{P}\{ e_{s_i}\ |\ c\}$ and $\mathbb{P}\{ e_{f_i}\ |\ c\}$, $i\in [0,\ M_w]$ as follows.
\begin{equation}
\mathbb{P}\{ e_{s_i}\ |\ c\} = \sum_{\overline{N}}^{} \mathbb{P}\{ \overline{N}\ |\ c\} \cdot \mathbb{P}\{ e_{s_i}\ |\ \overline{N}\}
\label{eq:psi}
\end{equation}
\begin{equation}
\mathbb{P}\{ e_{f_i}\ |\ c\} = \sum_{\overline{N}}^{} \mathbb{P}\{ \overline{N}\ |\ c\} \cdot \mathbb{P}\{ e_{f_i}\ |\ \overline{N}\}
\label{eq:pfi}
\end{equation}
\subsection{Parallelization of ECC algorithm}\label{subs:parallel}
As shown in Figure \ref{fig:flowchart-model}, ECC continuously performs the following steps: 
\begin{enumerate}
\item extracts a chain $c:\ s_c=\{e_1,\ e_2,\ ...,\ e_m\}$ from $L_c$;
\item generates all the events $e_x$ that can occur after $e_m$.
\item $\forall e_x$, adds $c_x:\ s_{c_x}=\{e_1,\ e_2,\ ...,\ e_m,\ e_x\}$ to $L_c$.
\end{enumerate}
During the execution of the three steps above only information associated with chain $c$ is used by the algorithm. Hence, it is possible to speed-up the execution by allowing more threads to manage different chains in parallel. In section \ref{s:ris} we show, through experimental measurements, that the execution time of the ECC algorithm decreases almost {\em linearly} with the number of used threads. This drastically reduces the computation time needed to generate all the possible outcomes and makes the analysis of large networks computationally feasible.

\subsection{Derivation of performance metrics}\label{subs:pm}
When $L_c$ becomes empty, $F_c$ contains all chains $c$ representing possible outcomes of the CSMA/CA execution with a probability to occur greater than, or equal to, $\theta$. By using chains $c\in F_c$, we derive the following metrics:
\begin{itemize}
 \item \emph{Coverage} (C): portion of event space covered by chains in $F_c$; it characterizes the ECC accuracy.
 \item \emph{Packet delivery ratio} (R): fraction of data packet successfully transmitted to the coordinator node; it measures the reliability provided by CSMA/CA.
 \item \emph{Packet latency} (L): delay experienced by a sensor node to successfully transmit its data packet to the coordinator node; it indicates the timeliness allowed by CSMA/CA in reporting an event.
 \item \emph{Energy consumption} (E): average total energy consumed by all sensor nodes in the network to report an event to the coordinator node; it measures the energy efficiency of CSMA/CA.
\end{itemize}
Let $c_i:\ s_{c_i}=\{e_1,\ e_2,\ ...,\ e_m \}$ be a specific chain in set $F_c$, and $p_{c_i}$ its associated probability. Then, the coverage $C$ can be calculated as follows:
\begin{equation}
C = \sum_{c_i\in F_c}^{} p_{c_i} 
\end{equation}

To derive the delivery ratio $R$ we compute, for each chain $c_i$, the fraction of successful transmissions $R_i$ occurred in $c_i$. Since there are $N$ nodes in the network and each node transmits one data packet, $R_i$ can be easily calculated as $R_i= \sfrac{{N_s}^i}{N}$ where ${N_s}^i$ is the number of successful transmissions in $c_i$. Since each chain $c_i$ has probability $p_{c_i}$ to occur, the delivery ratio $R$ is given by
\begin{equation}
R= \sum_{c_i\in F_c}^{} \frac{p_{c_i}}{C} \cdot R_i
\end{equation}

To calculate the average latency $L$ we first derive the probability density function (PDF) of packet latency. Specifically, the probability $P(t)$ to successfully receive a data packet with a delay equal to $t$ is 
\begin{equation}
P(t) = \sum_{c_i\in F_c: \exists e_j\in s_{c_i} | T_j=S \wedge f_j=t}^{}p_{c_i}
 \label{eq:pdf_pl}
\end{equation}
$P(t)$ is the sum of the probabilities of all chains $c_i$ containing a successful transmission at $t$. Then, by denoting as $t_{max}$ the largest instant at which a successful transmission can occur, we compute $L$ as:
\begin{equation}
L = \frac{\sum_{t=0}^{t_{max}} t\cdot P(t)}{\sum_{t=0}^{t_{max}}P(t)}
\end{equation}

Finally, we calculate the average total energy consumption $E$. Let $en_{c_i}$ be the total energy consumed by all sensor nodes when the events of chain $c_i$ occur. Hence,
\begin{equation}
E= \sum_{c_i\in F_c}^{} \frac{p_{c_i}}{C} \cdot en_{c_i} 
\end{equation}
The derivation of $en_{c_i}$ is reported in \cite{DomTesi}.

\section{Results}
\label{s:ris}
\begin{figure*}[!ht]
\minipage{0.33\textwidth}
\centering
\includegraphics[width=0.6\textwidth, angle=-90]{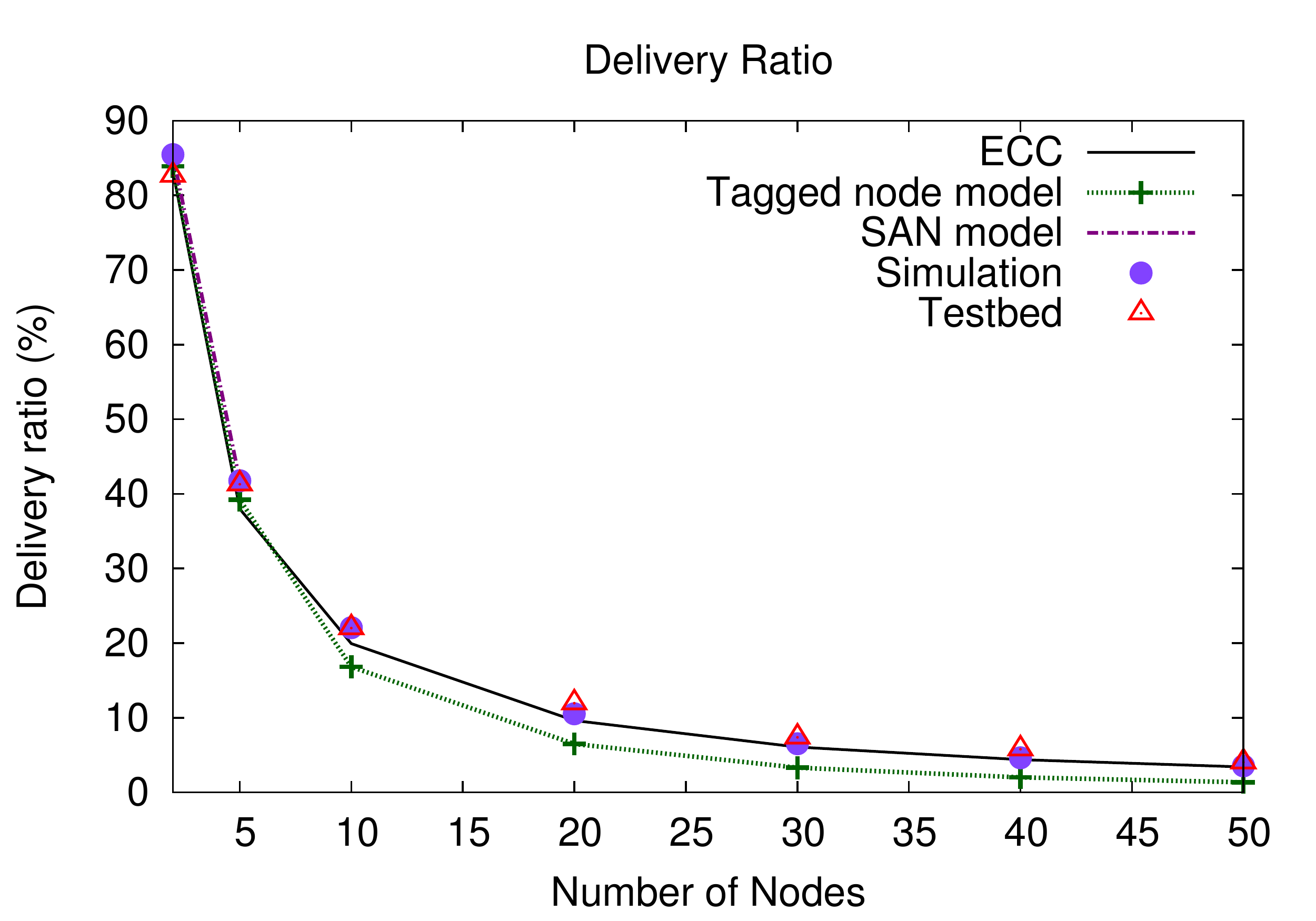}
\caption{Delivery Ratio.}
\label{fig:1}
\endminipage\hfill
\minipage{0.33\textwidth}
\centering
\includegraphics[width=0.6\textwidth, angle=-90]{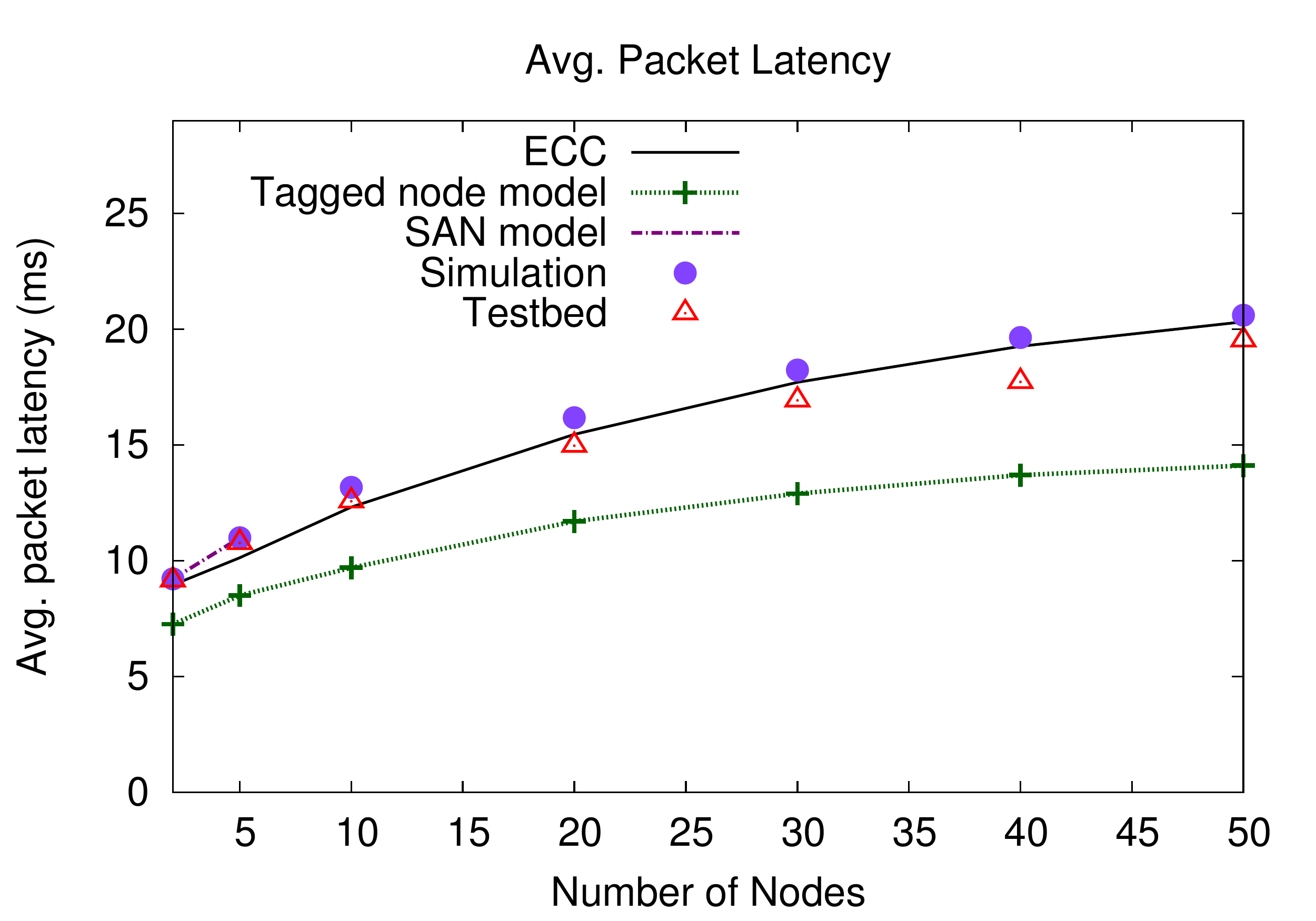}
\caption{Average Packet Latency.}
\label{fig:2}
\endminipage\hfill
\minipage{0.33\textwidth}%
\centering
\includegraphics[width=0.6\textwidth, angle=-90]{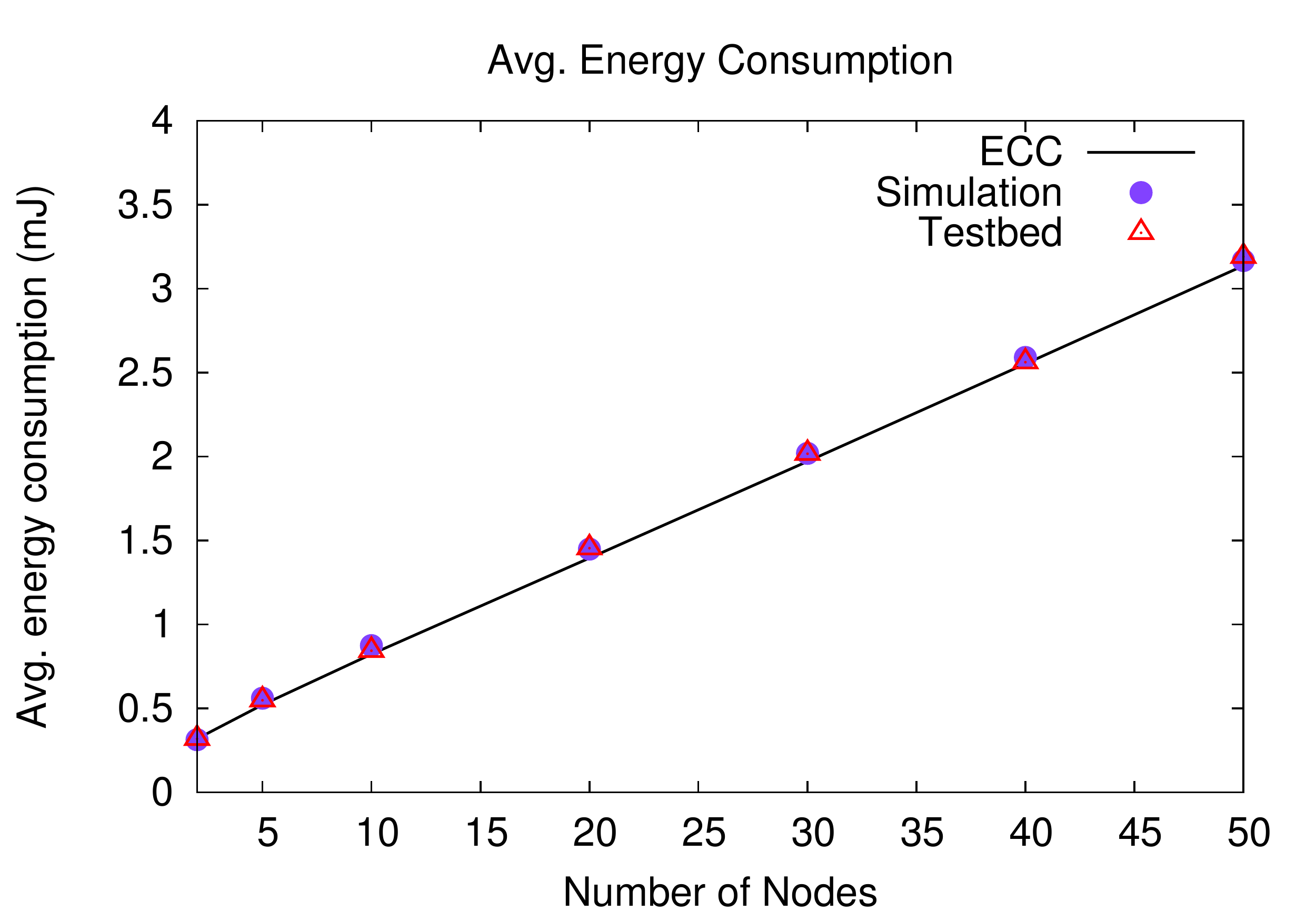}
\caption{Average Energy Consumption.}
\label{fig:3}
\endminipage
\end{figure*}

\begin{figure*}[!ht]
\minipage{0.33\textwidth}
\centering
\includegraphics[width=0.6\textwidth, angle=-90]{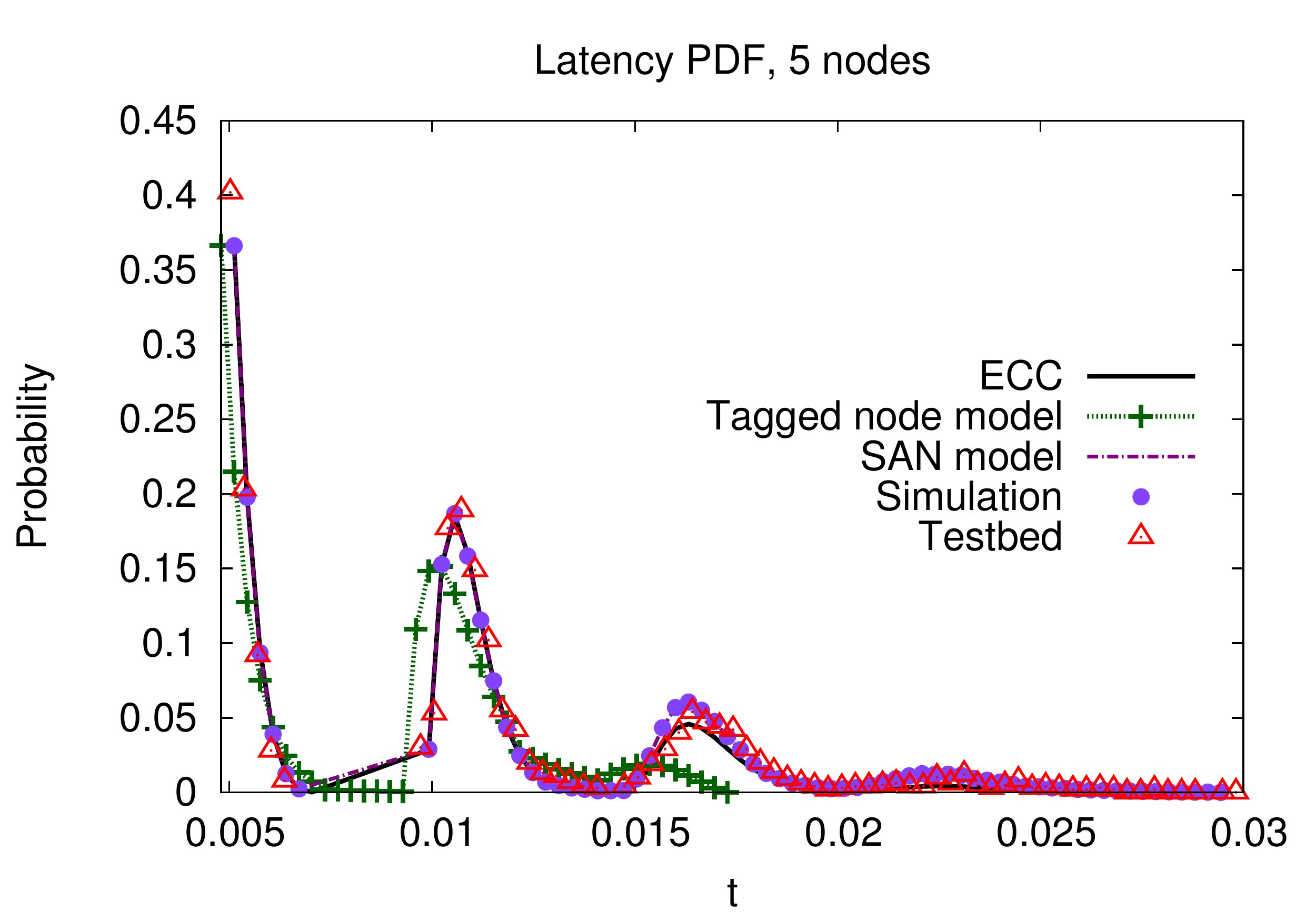}
\caption{Latency PDF, 5 nodes.}
\label{fig:4}
\endminipage\hfill
\minipage{0.33\textwidth}
\centering
\includegraphics[width=0.6\textwidth, angle=-90]{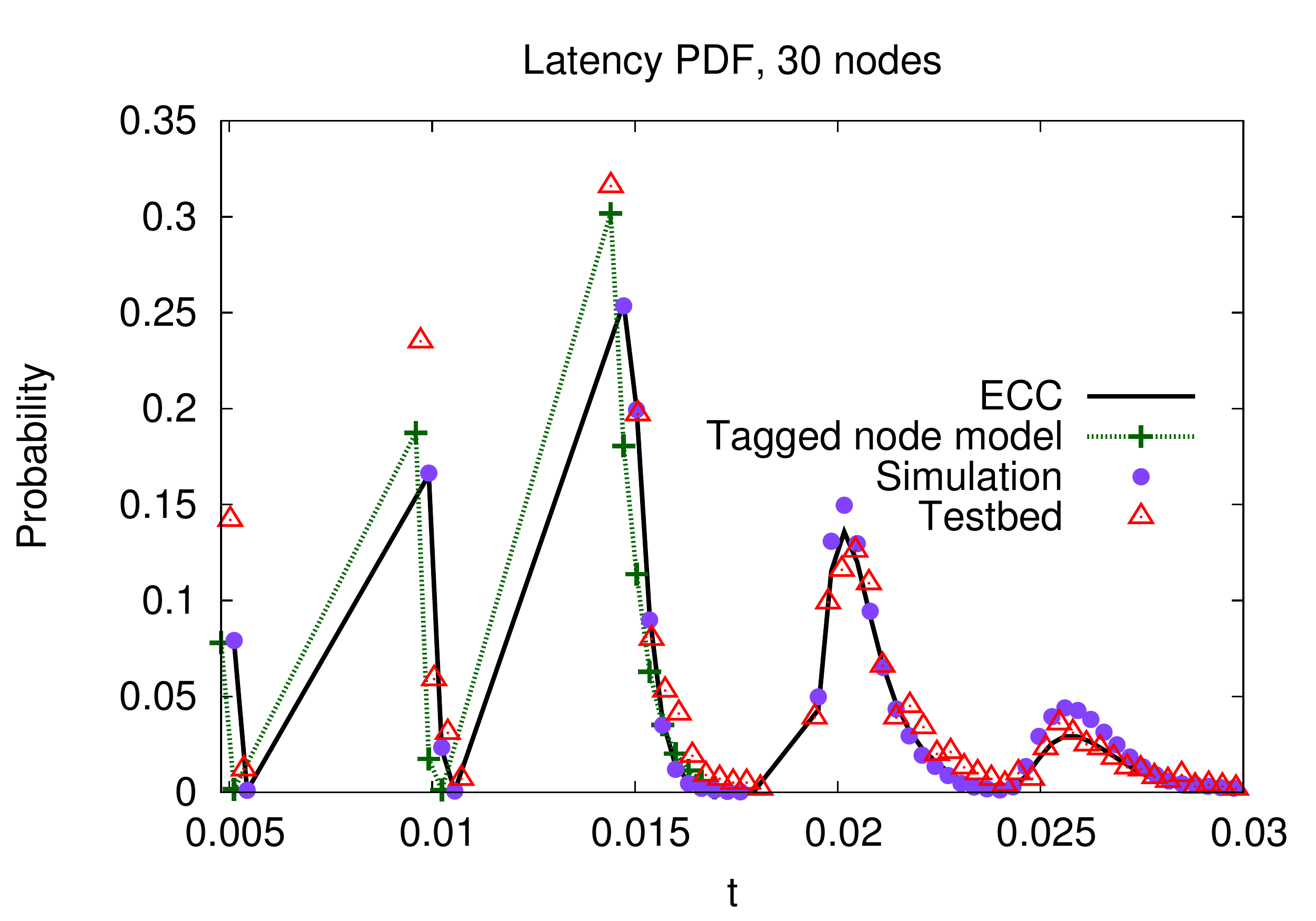}
\caption{Latency PDF, 30 nodes.}
\label{fig:5}
\endminipage\hfill
\minipage{0.33\textwidth}%
\centering
\includegraphics[width=0.6\textwidth, angle=-90]{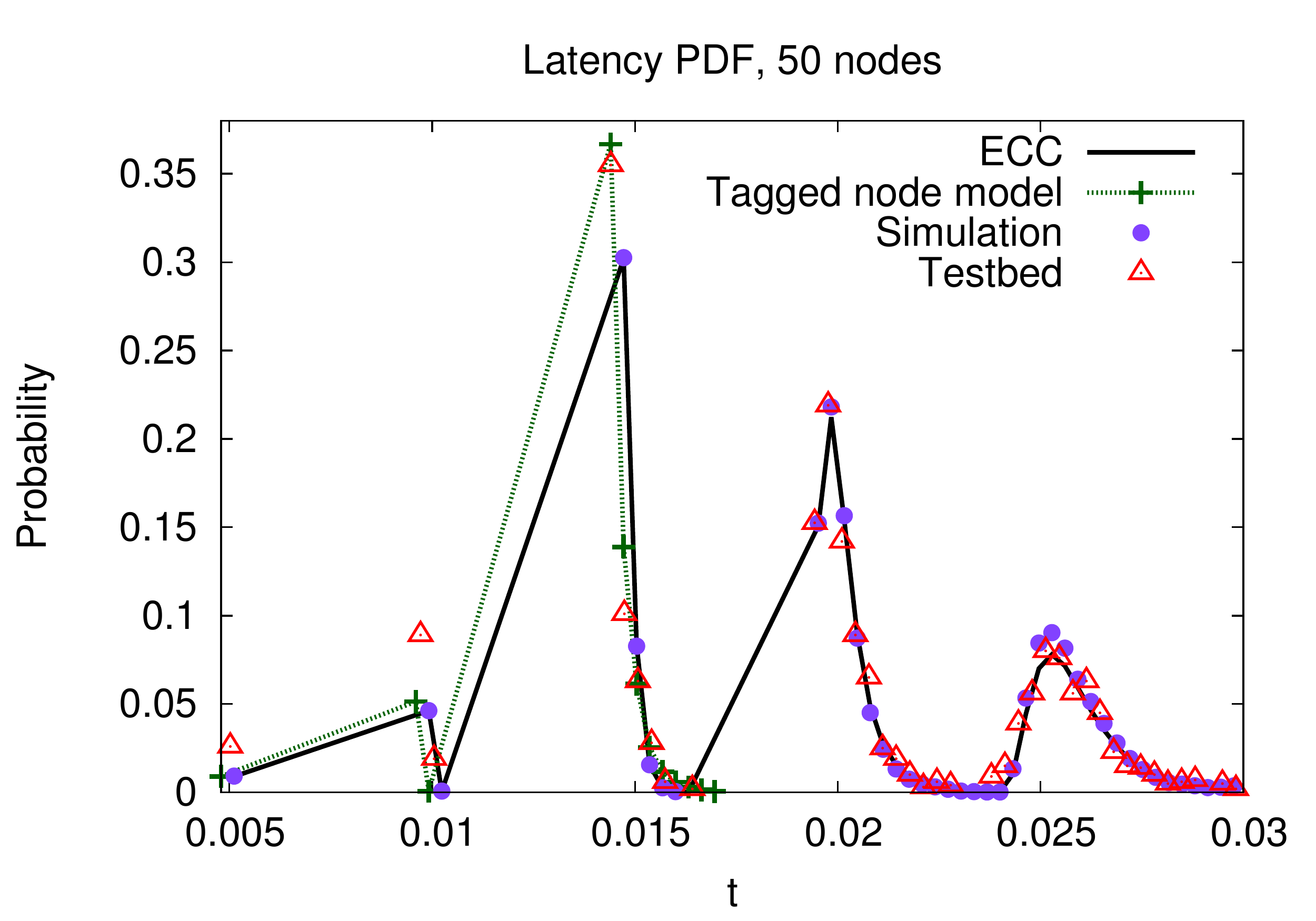}
\caption{Latency PDF, 50 nodes.}
\label{fig:6}
\endminipage
\end{figure*}
\begin{table*}[htb]
\centering
\caption{Analytical results, R(\%), L (ms), E (mJ), Time (s), with different values of $\theta$.}
\scriptsize
\begin{tabular}{|c | c | c | c | c | c | c |c | c | c | c | c | c | c | c | c | c | c | c |}
\hline
\multirow{2}{*}{N} & \multicolumn{6}{|c|}{$\theta = 0$} & \multicolumn{6}{|c|}{$\theta = 10^{-7}$} & \multicolumn{6}{|c|}{$\theta = 10^{-5}$}  \\ 
 & C & Chains &R& L\tiny(ms) & E\tiny(mJ) & Time\tiny(s) & C & Chains &R& L\tiny(ms) & E\tiny(mJ) & Time\tiny(s) & C & Chains &R& L\tiny(ms) & E\tiny(mJ) & Time\tiny(s)\\
\hline 
\tiny 10 & \tiny 1	&	\tiny 1842355	&	\tiny 19.94 &	\tiny 12.33 &	\tiny 0.822	& \tiny 14407 & \tiny 0.999	& \tiny 27590	& \tiny 19.94 & \tiny 12.32	& \tiny 0.822	& \tiny 398 & \tiny 0.970	& \tiny 3814	& \tiny 19.88	& \tiny 12.17	& \tiny 0.811	& \tiny 93 \\
\tiny 30 & \tiny 1	&	\tiny 1842355	&	\tiny 6.07	&	\tiny 17.81 &	\tiny 1.970	& \tiny 24768 & \tiny 0.999	& \tiny 19536	& \tiny 6.07	 & \tiny 17.81	& \tiny 1.970	& \tiny 922 & \tiny 0.978	& \tiny 3224	& \tiny 6.04	& \tiny 17.75	& \tiny 1.966	& \tiny 267 \\
\tiny 50 & \tiny 1	&	\tiny 1842355	&	\tiny 3.40	&	\tiny 20.36  & \tiny 3.130	& \tiny 33264 & \tiny 0.999	& \tiny 11509	& \tiny 3.40	& \tiny 20.36	& \tiny 3.130 & \tiny 875 & \tiny 0.987	& \tiny 2253	& \tiny 3.39	& \tiny 20.34	& \tiny 3.136	& \tiny 316\\
\hline
\end{tabular}
\label{tab:rel2}
\end{table*}
In this section, we evaluate the accuracy and tractability of our analytical model. We also use the model to investigate the performance of CSMA/CA in the considered event-driven scenario. Furthermore, we compare the results obtained using the ECC algorithm with those of the following two analytical models:
\begin{itemize}
\item \emph{Tagged node model} \cite{bur}. It focuses on a single node of the network (\emph{tagged-node}), and provides delivery ratio, average packet latency and delay distribution. The effects of acknowledgments and retransmissions are not considered.
\item \emph{SAN model} \cite{griba}. It makes use of stochastic automata networks (SANs) and considers the entire WSNs (i.e. all nodes altogether). It provides delivery ratio, average packet latency and delay distribution but not energy consumption. 
\end{itemize}
These models are the state-of-the-art models for event-driven WSNs using the unslotted 802.15.4 CSMA/CA. 

\subsection{Model validation}
We validate our analytical model through simulations implemented on the ns2 simulation tool \cite{ns2}, which includes the 802.15.4 module. In our simulations we consider a star network topology with sensor nodes located in a circle of radius $10m$ centered at the coordinator node. The transmission range was set to $15m$, while the carrier sensing range was set to $30m$. In all the simulations we assume that the 802.15.4 MAC protocol is operating in NBE mode with a $250$ Kbps bit rate. Power consumption values are derived from the Chipcon CC2420 radio transceiver \cite{cc2420}. Unless stated otherwise, we use the following set of CSMA/CA parameter values: $macMinBE=3,\ macMaxBE=4,\ macMaxBackoffs=2,\ macMaxFrameRetries=1$. 

For each simulation experiment, we performed $10$ independent replications, each of which consists of $10000$ cycles. In each cycle, it is assumed that all sensor nodes have to transmit one packet. We derived confidence intervals by using the independent replications method and a $95\%$ confidence level. However, they are typically very small and cannot be appreciated in the plots below.

Figures \ref{fig:1}-\ref{fig:3} compare simulation results with the analytical results provided by the ECC algorithm, and the models in \cite{bur} and \cite{griba}, in terms of delivery ratio, average latency and average energy consumption, respectively, for different network sizes (i.e., number of sensor nodes). We anticipate that the graphs also show the results related to the experimental evaluation described in Section 7.4 (\emph{Testbed} label). Finally, Figures \ref{fig:4}-\ref{fig:6} compare the probability density function (PDF) of packet latency for three different network sizes (i.e, with $5$, $30$, and $50$ sensor nodes). Figure \ref{fig:3} reports the analytical results of ECC only since the other two considered models do not provide the energy consumption. Also, in all the presented graphs, results related to SAN model are only available for networks composed of at most $5$ nodes. This is due to the very high complexity of this model (that increases exponentially with the network size) that prevented us to run the analysis with higher network sizes. Finally, we want to point out that the analytical results of ECC have been obtained using $\theta=0$.

In general, we observe that ECC results and simulation results almost overlap in all the considered scenarios. We studied additional scenarios, with different parameter values, and observed a similar agreement. Some of these results are presented in section 7.3.

Focusing on the Tagged node model \cite{bur}, we observe a good match between analysis and simulations for small network sizes. Conversely, the gap between analysis and simulations becomes significant with high network sizes. This is because the impact of retransmissions, which was not considered in \cite{bur}, becomes more apparent as the number of nodes increases (i.e. when the collision probability is high). Specifically,  by looking at Figures \ref{fig:4}-\ref{fig:6} we can observe that the percentage of packets with high delays significantly increases with higher network sizes. However, the Tagged node model is not able to capture at all the contribution of such long-latency packets (since it neglects retransmissions). Finally, we can observe a perfect match between the SAN model results and simulations. However, we remark that the high complexity of the model allows to study only very small networks. Also, we recall that, unlike ECC, the considered models do not provide the energy consumption of nodes, that is a fundamental metric for WSNs.

Figure \ref{fig:1} shows that the delivery ratio drastically reduces with the number of sensor nodes and it is very low even with a limited number of sensor nodes. This is because the 802.15.4 CSMA/CA is not able to manage contention efficiently even when the number of contending nodes is relatively low, due to its random nature. In the considered $event-driven$ scenario this limitation is more apparent as $all$ sensor nodes in the network start reporting data $simultaneously$, upon detecting an event.

Figures \ref{fig:2} and \ref{fig:3} show that both the average latency and energy consumption increase with the network size, as expected. This is because when the number of sensor nodes contending for channel access increases, the collision probability increases as well. Hence, sensor nodes consume more energy. In addition, they take more time to transmit their packets. Specifically, successful transmissions tend to occur some time after the beginning, when the level of contention is lower because some nodes have already given up, due to exceeded number of backoff stages or retransmissions (see Section 3). This behavior is confirmed by Figures \ref{fig:4}-\ref{fig:6}, where the highest spikes in the latency PDF shift to the right side of the graphs as the number of sensor nodes increases. 

\subsection{Impact of $\theta$ and multithreading}
In order to evaluate the trade-off between {\em tractability} and {\em accuracy} of our model, Table \ref{tab:rel2} summarizes the analytical results obtained with different values of $\theta$ (the network parameter values are the same as in Section 7.1). Specifically, Table \ref{tab:rel2} shows the number of chains generated by the ECC algorithm, the coverage of the event space ($C$), the performance metrics defined in Section 6.5 (i.e., $R$, $L$ and $E$) and, finally, the computation time (in seconds) required to compute the same performance metrics. When using $\theta = 0$, a coverage of 100\% and, hence, the maximum accuracy of results is obtained. However, the number of generated chains is very large (more than 1800000), which requires a high computation time (33264s in the case of 50 sensor nodes, i.e. 9 hours). 

As expected, the coverage of the event space decreases as the value of $\theta$ increases. With $\theta = 10^{-5}$, it reduces to $97-98\%$, however the model still obtains nearly the same accuracy as with $\theta = 0$. Moreover, the computation time reduces by a factor larger than $100$, and the number of generated chains reduces by a factor larger than $500$, with respect to $\theta = 0$, when $N = 50$. These results can be explained as follows. As mentioned in Section 6, ECC analyzes only events that are most likely to occur in the network, which are the events that will influence most the performance metrics. This way, ECC achieves approximately the same accuracy as when $\theta = 0$ while analyzing a much lower number of events. Therefore, it saves a tremendous amount of computation time.
 
Surprisingly, Table \ref{tab:rel2} shows that, when $\theta > 0$, the number of chains generated by ECC \emph{decreases} as the network size increases. In addition, the coverage {\em increases} although the number of chains decreases. This apparently counterintuitive behavior can be explained as follows. When the number of sensor nodes increases, some events become significantly more likely to occur than others. For instance, if $N = 50$, collisions are more likely to occur than successful transmissions. Therefore, when $N = 50$, ECC will select, at each step, fewer (yet very likely) events with respect to the case $N = 5$. Hence, the number of considered chains decreases (and the coverage increases), as the number of nodes increases. This property of ECC reduces the computation needed to calculate the metrics of interest when the network size increases, making the analysis of large networks easier.

Finally, to further assess the model tractability, we measured the average computation time of the analytical model as a function of the number of threads that are activated. Figure \ref{fig:parall} shows that the computation time decreases almost {\em linearly} with the number of threads. This is because the algorithm is implemented in such a way to assign the computation of the various chains to different threads. Since threads synchronize only to modify a global data structure (i.e. the list of chains $L_c$), each thread is almost independent from the others. Hence, the almost linear decrease of computation time. 
\subsection{Impact of CSMA/CA parameters}
\begin{figure*}[tbp]
\minipage{0.25\textwidth}
\centering
\includegraphics[width=1\textwidth]{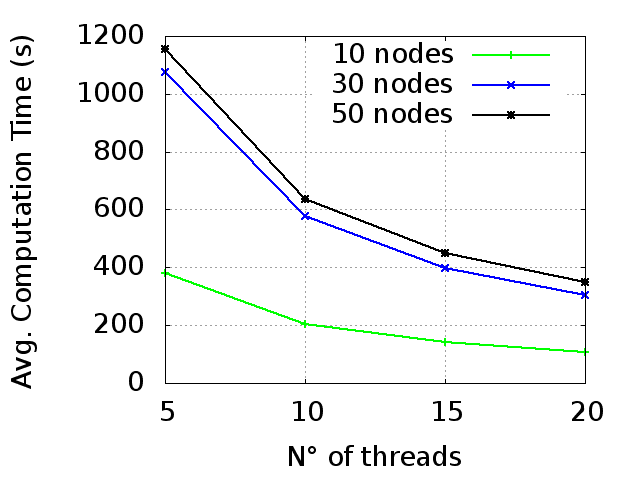}
\caption{Avg. computation time vs. number of threads.}
\label{fig:parall}
\endminipage\hfill
\minipage{0.25\textwidth}
\centering
\includegraphics[width=0.73\textwidth, angle=-90]{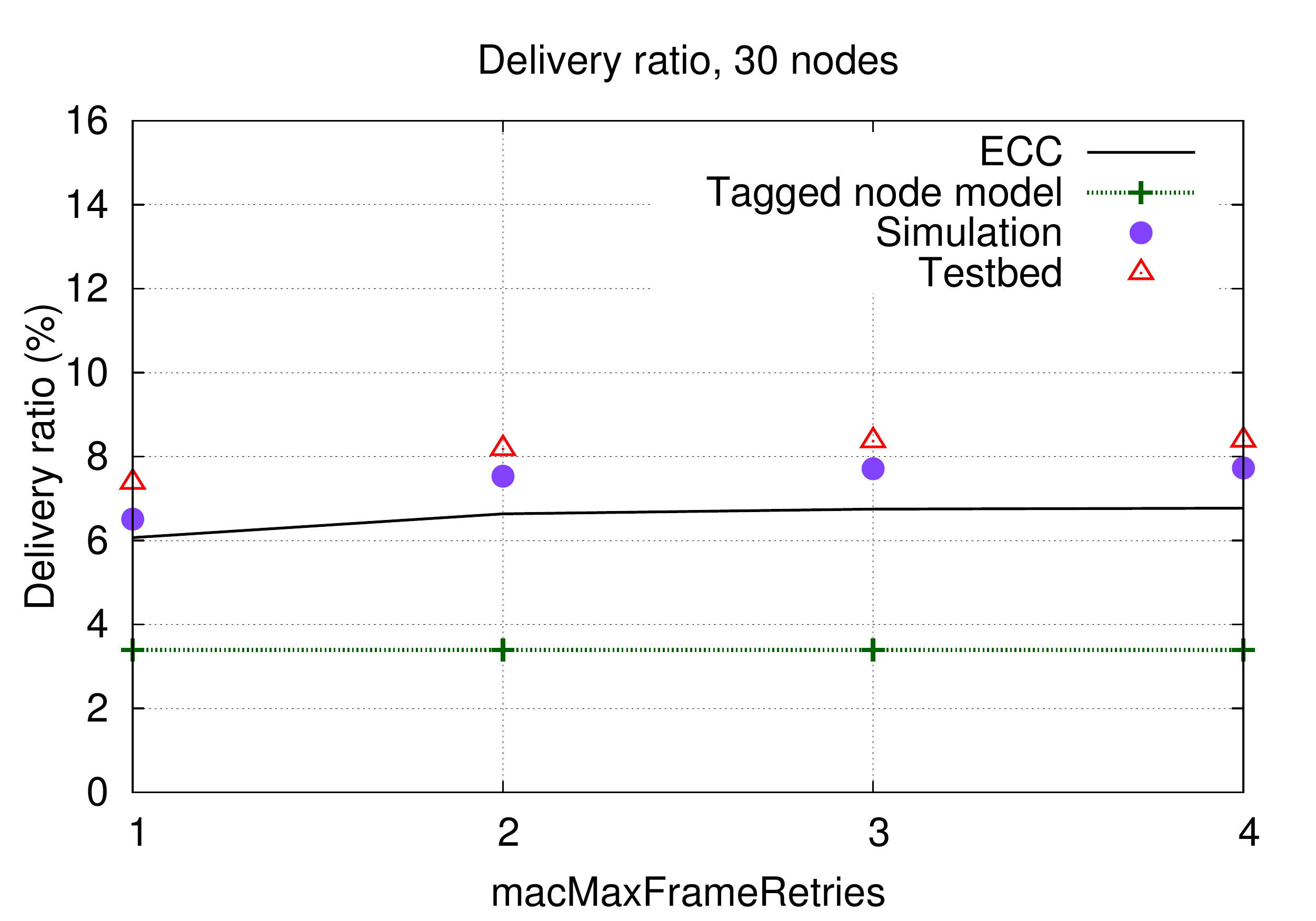}
\caption{Delivery ratio vs. $macMaxFrameRetries$}
\label{fig:7}
\endminipage\hfill
\minipage{0.25\textwidth}
\centering
\includegraphics[width=0.73\textwidth, angle=-90]{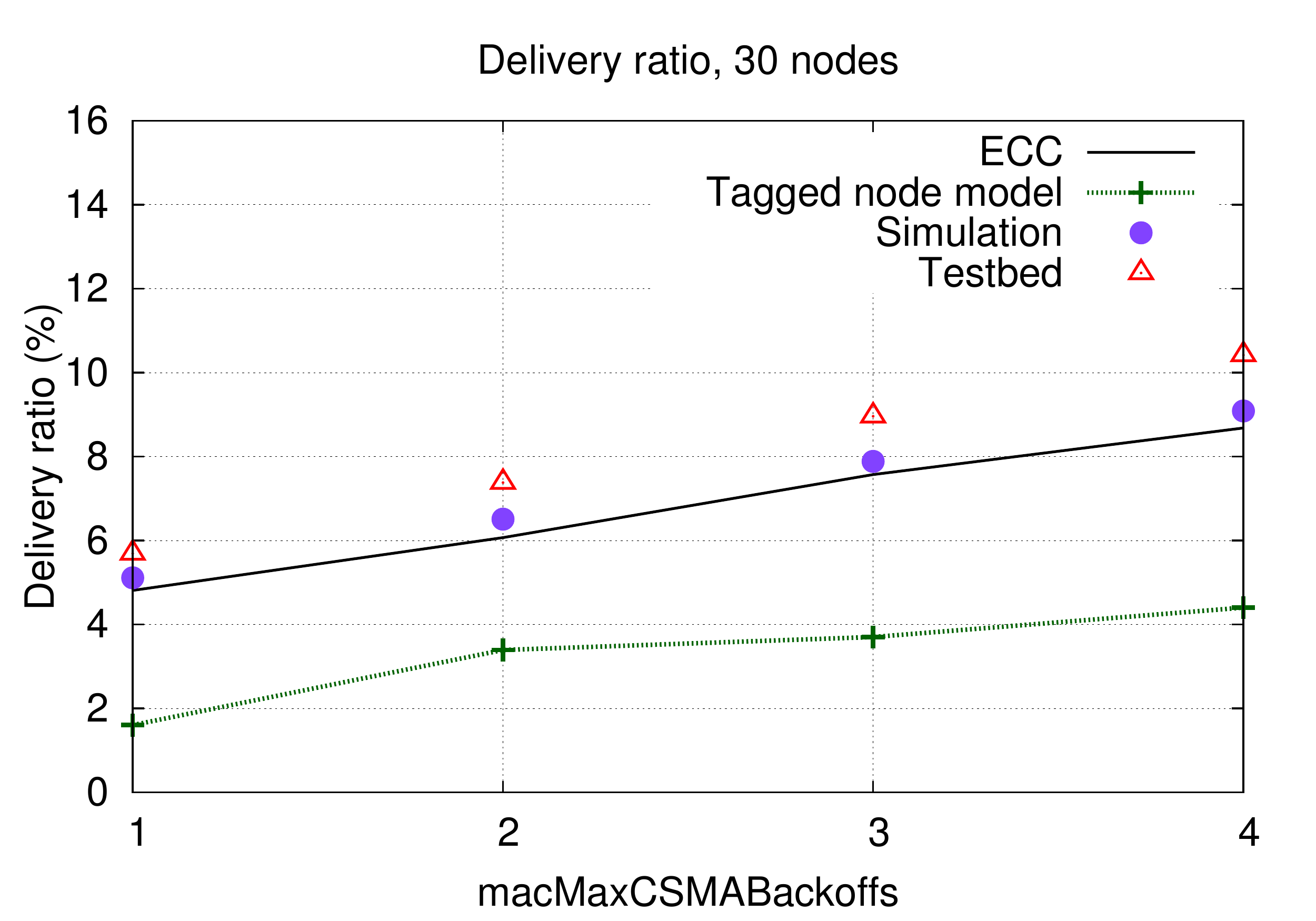}
\caption{Delivery ratio vs. $macMaxCSMABackoffs$}
\label{fig:8}
\endminipage\hfill
\minipage{0.25\textwidth}%
\centering
\includegraphics[width=0.73\textwidth, angle=-90]{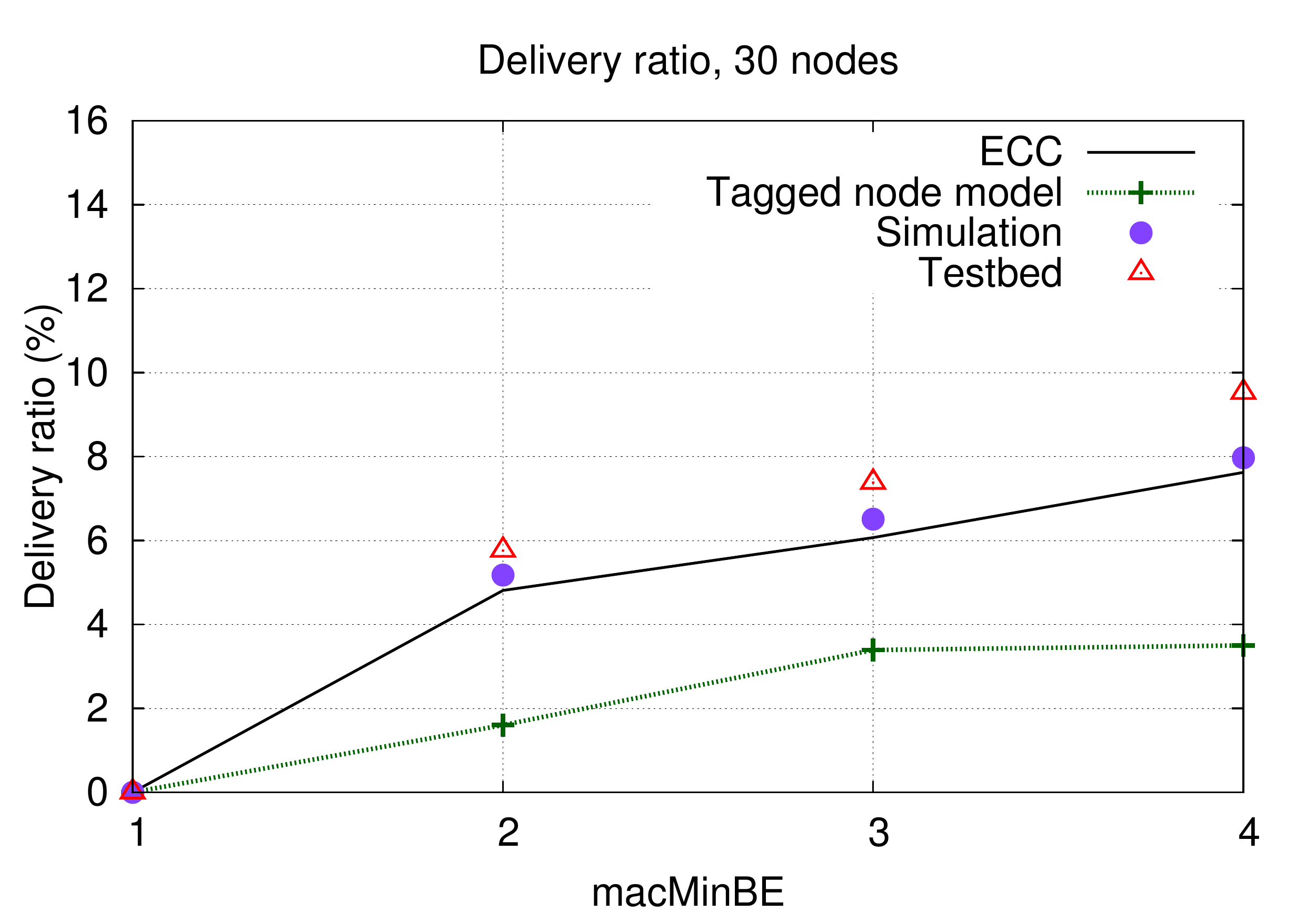}
\caption{Delivery ratio vs. $macMinBE$}
\label{fig:9}
\endminipage
\end{figure*}
In this section we evaluate the impact of each single CSMA/CA parameter on the overall performance. To this end, we focus on a network with $30$ sensor nodes. For the sake of space, we only show the delivery ratio, however, we also derived the average latency and energy consumption (see \cite{DomTesi} for the complete set of results). We also report the analytical results obtained by Tagged node model. We do not consider the SAN model due to its high complexity with the considered network size. 

In general, we can observe a quite good match between ECC results and simulation results. Conversely, there is a significant gap between simulation results and the Tagged node model results. This is due to the semplifications introduced by this model.

Figure \ref{fig:7} shows that increasing the maximum allowed number of retransmissions (i.e. $macMaxFrameRetries$) does not provide any significant effect on the delivery ratio, for values larger than one. This is because the majority of packets are dropped by the MAC protocol because of exceeded number of backof stages (i.e., consecutive CCAs). Indeed, increasing the maximum number of backoff stages (i.e. the $macMaxBackoffs$) results in an increase of the delivery ratio, as shown in Figure \ref{fig:8}. Furthermore, we also observed an increase in both the average latency and energy consumption. This is because a larger number of packets is successfully transmitted, which takes more time and more energy. Finally, Figure \ref{fig:9} shows that increasing the minimum backoff window size (i.e. $macMinBE$) also causes an increase of the delivery ratio. The motivation is that a larger initial backoff window size reduces the collision probability at the first backoff stage, thus increasing the probability of successful transmission. Again, we observed an increase in the average latency. Instead, the energy consumption decreases since less collisions occur.

The above results show that increasing the CSMA/CA parameter values is beneficial in terms of increased communication reliability. However, it also increases the energy consumption and packet latency, which may not be good for time-critical applications. Hence, the most appropriate parameter setting depends on the specific application scenario. 
\subsection{Experimental evaluation}
To verify the ability of our analytical model to predict the performance of the WSN in a real environment, we implemented the 802.15.4 unslotted CSMA/CA algorithm in the Contiki OS \cite{contiki}. For experimental measurements we used a testbed composed of Tmote-sky motes \cite{tmote}. As in the analysis, all nodes simultaneously start transmitting a data packet to a common coordinator node. We repeat this experiment 1000 times.

We performed different sets of experiments varying the number of nodes in the network (from $2$ up to $50$) as well as the MAC parameter values. The obtained results are reported in Figures \ref{fig:1}-\ref{fig:6} and Figures \ref{fig:7}-\ref{fig:9} (\emph{Testbed} label). In general, we can observe that experimental and ECC results are very close in all the considered scenarios. In some cases we noticed that the performance of the real WSN is slightly better than what predicted by the analysis (e.g. lower latency in experiments with respect to analysis (Figure \ref{fig:2}) and higher spikes of PDF in experiments with respect to analysis (Figures \ref{fig:4}-\ref{fig:6})). We found that this apparently counterintuitive behavior is due to the \emph{capture effect} that occurs in the real testbed, i.e. the ability of the radio to correctly receive a strong signal from one transmitter, despite significant interference from other transmitters \cite{refCE}.
\section{Conclusions}
In this paper, we have presented an analytical model of the unslotted CSMA/CA algorithm used in 802.15.4 WSNs operating in NBE mode. The proposed model is both accurate and efficient. It leverages an approach called \emph{Event Chains Computation} (ECC), that reduces complexity, with a limited loss of accuracy, by removing event sequences whose probability to occur is below a given threshold. We have shown that the computation time required for deriving the performance metrics of interest can be reduced by a factor larger than 100, or more, with a negligible impact on the accuracy of the obtained results. Also, our model can exploit a multi-threading approach, thus taking advantage of a parallel execution. Our model allows an accurate analysis of 802.15.4 WSNs in NBE mode, even when there is a large number of sensor nodes. As a matter of fact, we have used our model to investigate the impact of different CSMA/CA parameters on the WSN performance. We have observed that, in the considered scenario, the delivery ratio is very low even when the number of contending nodes is relatively low. It can be improved by increasing the initial backoff-window size and/or the maximum number of backoff stages allowed for each packet. However, this also increases the energy consumption and packet latency.
\section{Acknowledgments}
We thank the AE and the reviewers for their insightful comments.

\vspace*{-1.31cm}
\begin{IEEEbiographynophoto}{Domenico De Guglielmo}
is a Postdoctoral Researcher in the Dept. of Information Engineering at the University of Pisa. His research interests are in the field of WSNs and Internet of Things.
\end{IEEEbiographynophoto}
\vspace*{-1.32cm}
\begin{IEEEbiographynophoto}{Francesco Restuccia}
is a Ph.D candidate in the Department of Computer Science at the Missouri University of Science and Technology. His research
interests include pervasive and mobile computing and WSNs.
\end{IEEEbiographynophoto}
\vspace*{-1.32cm}
\begin{IEEEbiographynophoto}{Giuseppe Anastasi}
is a Full Professor at the Dept. of Information Engineering at the University of Pisa, Italy. He is also the Director of the CINI National Smart Cities Lab. His research interests include pervasive computing, sensor networks, sustainable computing, and ICT or smart cities. He has contributed to many research programs funded by both national and international institutions. He has co-edited two books and published more than 120 papers in international journal and conferences. Dr. Anastasi is an Associate Editor of Sustainable Computing, and Pervasive and Mobile Computing. He has served as General Co-chair and Program Chair of many international conferences.
\end{IEEEbiographynophoto}
\vspace*{-1.33cm}
\begin{IEEEbiographynophoto}{Marco Conti}
is Research Director of the Italian National Research Council (CNR) and Director of the CNR department of Engineering, ICT and Technologies for Energy and Transports. He published in journals and conference proceedings more than 300 papers related to design, modelling, and performance evaluation of computer-network architectures and protocols. He coauthored the book Metropolitan Area Networks (1997) and is co-editor of the books Mobile Ad Hoc Networking (2004), Mobile Ad Hoc Networks: From Theory to Reality (2007), and Mobile Ad Hoc Networking: Cutting Edge Directions (2012). He is the Editor-in-Chief of Elsevier's Computer Communications Journal, and Associate Editor-in-Chief of Elsevier's Pervasive and Mobile Computing (PMC) Journal. He served as general/program chair for several conferences, including IEEE PerCom, IEEE WoWMoM, IEEE MASS, ACM MobiHoc, and IFIP TC6 Networking. 
\end{IEEEbiographynophoto}
\vspace*{-1.33cm}
\begin{IEEEbiographynophoto}{Sajal K. Das}
is the Chair and Daniel St. Clair Chair Professor of Computer Science Department at the Missouri University of Science and Technology. During 2008-2011, he served the US National Science Foundation as a Program Director in the division of Computer Networks and Systems. His current research interests include wireless and sensor networks, mobile and pervasive computing, smart environments and health care, security and privacy, distributed and cloud computing, biological and social networks, applied graph theory and game theory. He has published over 500 papers in journals and conferences, 47 book chapters, and holds five US patents. He has also coauthored three books. Dr. Das is a recipient of the IEEE Computer Society Technical Achievement Award for pioneering contributions in sensor networks and mobile computing. He is the Founding Editor-in-Chief of Elsevier's Pervasive and Mobile Computing (PMC) journal, and an Associate Editor of IEEE Transactions on Mobile Computing, ACM Transactions on Sensor Networks, ACM/Springer Wireless Networks, Journal of Parallel and Distributed Computing, and Journal of Peer-to-Peer Networking and Applications.
\end{IEEEbiographynophoto}
\end{document}